\documentclass[a4paper,reqno]{amsart}
\pdfoutput=1

\usepackage[utf8]{inputenc}
\usepackage{amsmath,amssymb,amsfonts,amsthm}
\usepackage{tikz}
\usetikzlibrary{arrows}
\usetikzlibrary{arrows.meta}
\usetikzlibrary{snakes}
\usepackage[normalem]{ulem}
\usepackage{wrapfig}
\usepackage{url}
\usepackage[unicode=true,bookmarks=false,colorlinks=true]{hyperref} 
\usepackage[capitalise]{cleveref}

\usepackage{enumerate}
\usepackage[shortlabels]{enumitem}

\usepackage[yyyymmdd,hhmmss]{datetime}

\usepackage{mathtools}
\usepackage{stmaryrd}
\usepackage{thmtools}
\usepackage{float} 

 \usepackage{relsize}

\theoremstyle{plain}
\newtheorem{theorem}{Theorem}
\newtheorem*{theorem*}{Theorem}
\newtheorem{corollary}[theorem]{Corollary}
\newtheorem{lemma}[theorem]{Lemma}

\theoremstyle{definition}
\newtheorem{remark}[theorem]{Remark}
\newtheorem{example}[theorem]{Example}
\newtheorem{definition}[theorem]{Definition}

\newcommand{\UU}{\mathcal{U}}

\newcommand{\refl}{\mathsf{refl}}

\newcommand{\ct}{%
  \mathchoice{\mathbin{\raisebox{0.5ex}{$\displaystyle\centerdot$}}}%
             {\mathbin{\raisebox{0.5ex}{$\centerdot$}}}%
             {\mathbin{\raisebox{0.25ex}{$\scriptstyle\,\centerdot\,$}}}%
             {\mathbin{\raisebox{0.1ex}{$\scriptscriptstyle\,\centerdot\,$}}}
}

\newcommand{\istype}[1]{\mathsf{is}\mbox{-}{#1}\mbox{-}\mathsf{type}}
\newcommand{\trunc}[2]{\mathopen{}\left\Vert #2\right\Vert_{#1}\mathclose{}}

\newcommand{\tproj}[3][]{\mathopen{}\left|#3\right|_{#2}^{#1}\mathclose{}}

\newcommand{\defeq}{\vcentcolon\equiv}

\newcommand{\glue}{\ensuremath{\mathsf{glue}}}

\newcommand{\inl}{\ensuremath{\mathsf{inl}}}
\newcommand{\inr}{\ensuremath{\mathsf{inr}}}

\newcommand{\bool}{\mathbf 2}
\newcommand{\N}{\mathbb N}

\newcommand{\List}{\mathsf{List}}
\newcommand{\unit}{\mathbf 1}

\newcommand{\quotXY}[2]{{#1} \! \mathop{\sslash} \! {#2}}
\newcommand{\quotX}[1]{\quotXY{#1}{\scriptstyle \mathop{\leadsto}}}
\newcommand{\quot}{\quotX{A}}

\newcommand{\setquotXY}[2]{{#1} \! \mathop{\slash} \! {#2}}
\newcommand{\setquotX}[1]{\setquotXY{#1}{\scriptstyle \mathop{\leadsto}}}
\newcommand{\setquot}{\setquotX{A}}

\newcommand{\setquotsim}{A \! \mathop{\slash} \! {\scriptstyle \mathop{\sim}}}

\newcommand{\ap}{\ensuremath{\mathsf{ap}}}

\newcommand{\WB}{\mathsf{WB}}
\newcommand{\LC}{\mathsf{LC}}
\newcommand{\CS}{\mathsf{C}}
\newcommand{\CR}{\mathsf{CR}}
\newcommand{\RINV}{\mathsf{RINV}}
\newcommand{\LINV}{\mathsf{LINV}}
\newcommand{\INV}{\mathsf{INV}}
\newcommand{\wb}{\mathsf{wb}}

\newcommand{\listext}[1]{{#1}_L}

\newcommand{\leadstostar}{\leadsto^{*}} 
\newcommand{\leadstofrom}{\leftrightsquigarrow}
\newcommand{\leadsfrom}{\mathrel{\reflectbox{$\leadsto$}}}
\newcommand{\leadsfromstar}{\mathrel{\reflectbox{$\leadsto^{*}$}}}
\newcommand{\leadstofromstar}{\leftrightsquigarrow^{*}} 

\newcommand{\torel}{\mathop{\leadsto}}

\newcommand{\from}{\mathop{\leadsfrom}}
\newcommand{\tofrom}{\mathop{\leadstofrom}}
\newcommand{\tostar}{\mathop{\leadstostar}}
\newcommand{\fromstar}{\mathop{\leadsfromstar}}
\newcommand{\tofromstar}{\mathop{\leadstofromstar}}
\newcommand{\tofromzero}{\tofrom^0}
\newcommand{\tozero}{\torel^0}

\newcommand{\To}{\mathop{\Rightarrow}}
\newcommand{\Tostar}{\mathop{\xRightarrow{*}}}

\newcommand{\Tofrom}{\mathop{\Leftrightarrow}}
\newcommand{\Tofromstar}{\mathop{\xLeftrightarrow{*}}}

\newcommand{\acc}[1]{\mathsf{acc}^{\mathsmaller{#1}}}
\newcommand{\cons}{\cdot}
\newcommand{\consLong}{\mathsf{snoc}}
\newcommand{\length}{\mathsf{length}}

\newcommand{\hcolim}{\mathsf{hcolim}}
\DeclareFontFamily{U}{wasysmall}{}
\DeclareFontShape{U}{wasysmall}{m}{n}{
	<-5.5> s*[1.4] wasy5
	<5.5-6.5> s*[1.4] wasy6
	<6.5-7.5> s*[1.4] wasy7
	<7.5-8.5> s*[1.4] wasy8
	<8.5-9.5> s*[1.4] wasy9 
	<9.5-> s*[1.4] wasy10
}{}

\newcommand{\fhg}{\mathsf{F}}

\newcommand{\wasSigmaNull}{A}

\title[]{A Rewriting Coherence Theorem with Applications in Homotopy Type Theory}%
\thanks{Funding Acknowledgment: This work was supported by The Royal Society, grant reference URF\textbackslash R1\textbackslash 191055.}
\author{Nicolai Kraus \and Jakob von Raumer}

\begin{document}

\begin{abstract}
Higher-dimensional rewriting systems are tools to analyse the structure of formally
reducing terms to normal forms, as well as comparing the different reduction paths
that lead to those normal forms. This higher structure can be captured by finding
a \emph{homotopy basis} for the rewriting system.
We show that the basic notions of \emph{confluence} and \emph{wellfoundedness}
are sufficient to recursively build such a homotopy basis, with a construction
reminiscent of an argument by Craig C. Squier. We then go on to translate this
construction to the setting of \emph{homotopy type theory}, where managing
equalities between paths is important in order to construct functions which are
\emph{coherent} with respect to higher dimensions. Eventually, we apply the result
to approximate a series of open questions in homotopy
type theory,
such as the characterisation of the homotopy groups of the
\emph{free group on a set} and the \emph{pushout of 1-types}.

This paper expands on our previous conference contribution \emph{Coherence via Wellfoundedness} by laying out the construction in the language of higher-dimensional rewriting.
\end{abstract}

\maketitle

\tableofcontents

\section{Introduction} \label{sec:intro}

\subsection{Confluence and coherence in higher dimensional rewriting}
\label{subsec:intro-1}

In classical mathematics and computer science, a relation $\leadsto$ on a set $M$ is called \emph{terminating} or \emph{Noetherian}
if there is no infinite sequence $x_0 \leadsto x_1 \leadsto x_2 \leadsto \ldots$.
This is a standard property of term rewriting (or reduction) systems such as the \emph{typed lambda calculus} \cite{barendregt1992handbook} and ensures that any term  can be fully reduced, i.e.\ after finitely many steps an irreducible term (a \emph{normal form}) is reached. 
Similarly, \emph{confluence} is the property that, whenever one has $x \leadsfromstar w \leadstostar y$, there is a $z$ such that $x \leadstostar z \leadsfromstar y$. 
This property expresses that the final result of a sequence of reductions is independent of the order in which reduction steps are performed.
Together, termination and confluence guarantee that every term has a unique normal form.

We can go up one level:
Rather than asking whether two reduction sequences give the same result, we can ask whether two parallel reductions $u, v \in (x \leadstostar y)$ are equal, or related, in some appropriate sense; in other words, we can ask whether the system of reduction steps is \emph{coherent}, i.e.\ whether different steps ``fit together''.
One way to give meaning to the question is to consider (higher) rewriting steps between reduction sequences of the form $\alpha \in (u \To v)$ and identify a set of ``good'' steps.
If any reduction sequence $u$ can be rewritten into any sequence $v$ parallel to it by concatenating or pasting ``good'' steps or their inverses, 
the set of ``good'' reduction steps is called a \emph{homotopy basis}.

Termination and confluence of the reduction relation $\leadsto$ also help with the construction of homotopy bases.
As an example, Newman's lemma shows that two reduction sequences starting from a single object $x$ can be completed to parallel reduction sequences ending in the same object, and
Newman \cite{newman1942theories} as well as Huet \cite{Huet1980ConfluentRA} essentially demonstrate that the resulting shape can be filled up with smaller shapes which witness local confluence.
Going further, Newman \cite{newman1942theories} also shows that a closed zig-zag of reductions is always deformable into the empty cycle using local confluence diagrams,%
\footnote{We 
thank Vincent van Oostrom for pointing this out and for a variety of further helpful remarks; we also refer to the acknowledgements at the end of the paper.}
and a generalisation of this observation will be central to the current paper.

An argument in the same direction was presented by Squier \cite{squierOtto1987word,squier1987word}.
He constructed homotopy bases to show that a monoid with a finite presentation which is terminating and confluent satisfies the homological finiteness condition known as $\mathsf{(FP)_3}$. 
Squier's original application was that monoids can have a decidable word problem without having a finite, terminating, and confluent presentation.
A further improvement was made by Otto and Kobayashi using new ideas by Squier \cite{squier1994finiteness}.
Later, constructions based on Squier's arguments were used to solve various coherence problems, for example 
for Gray- \cite{forest_et_al:LIPIcs:2018:9185} and
monoidal categories \cite{guiraud2012coherence} as well as for
Artin \cite{gaussent2015coherent},
plactic \cite{hage2017knuth},
and general monoids \cite{guiraud_et_al:LIPIcs:2013:4064}.
The aim of the current paper is to show that a rewriting result, reminiscent of Newman's and Squier's, can also be applied in homotopy type theory.

\subsection{Higher-dimensional structures in type theory}
\label{subsec:intro-2}

Higher-dimensional structures also naturally appear in \emph{Martin-L\"of type theory}%
~\cite{martinlof84sambin}, 
where types are $\infty$-groupoids (globular~\cite{lumsdaine_weakOmegaCatsFromITT,bg:type-wkom}, simplicial~\cite{kraus_generaluniversalproperty,nicolai:thesis}, internal~\cite{finster_internal,ann-cap-kra:two-level}) and universes are $(\infty,1)$-categories (cf.~\cite{ann-cap-kra:two-level} $\infty$-groupoids).
Given a type $A$ in a type-theoretic universe $\UU$, terms $x,y : A$ are the objects (0-cells) of the higher structure determined by $A$.

What allows us to form the morphisms and higher morphisms is
Martin-L\"of's \emph{identity type}.
Given terms $x$ and $y$ of type $A$, we write this type as $\mathsf{Id}_A(x,y)$ or simply as $x = y$.
The type theory community calls an element (or a term) $p : x = y$
a \emph{(propositional) equality}, an \emph{identification}, or (and this is the terminology we use in this paper) a \emph{path}.
The path type $x = y$ is the type of 1-cells, and for $p, q : x = y$, the type $p = q$ (the iterated path type) is the type of 2-cells, and so on.
This view can be made precise further by interpreting type theory in the setting of abstract homotopy theory~\cite{awodeyWarren_HTmodelsOfIT} and especially in simplicial sets \cite{kap-lum:simplicial-model} or cubical sets~\cite{cchm:cubical}, themselves models of topological spaces.

\emph{Homotopy type theory}~\cite{hott-book} is a variation of Martin-L\"of type theory which embraces this interpretation of types as spaces by adding internal principles that are justified by a range of models.
The central such principle is Voevodsky's \emph{univalence axiom}~\cite{kap-lum:simplicial-model}, presented as an assumption in the original formulation of homotopy type theory (cf.~\cite{hott-book}) but derivable internally in cubical type theories (cf.~\cite{cchm:cubical}).

Most of the results establishing that types are $\infty$-groupoids are of meta-theoretic nature (with the exception of attempts to internalise meta-theoretic results~\cite{finster_internal,ann-cap-kra:two-level}).
When working internally, for example when using a proof assistant, the other side of the coin becomes visible:
Attempting to treat types as sets, or the universe as an ordinary 1-category, often leads to problems which stem from the fact that paths are structure rather than property (i.e.\ they are not unique).
In other words, a path $p : x = y$ should be seen as an isomorphism in a higher category rather than an equality between elements of a set.
A collection of (iso-) morphisms in a higher category often requires coherences in order to be well-behaved.
The same is the case in homotopy type theory, and constructing these coherences is often a central difficulty.

\subsection{Applying rewriting arguments in homotopy type theory}
\label{subsec:intro-3}

The goal of this paper is to demonstrate how techniques from higher-dimensional rewriting can be applied in homotopy type theory, and to derive several new and non-trivial results with these techniques.
Although the most immediate applications of rewriting to type theory are of meta-theoretic nature (e.g.\ showing that $\beta$-reduction is confluent), this is explicitly \emph{not} what we mean; as indicated in the previous paragraph above, we mean purely internal type-theoretic applications.

The arguments that we apply are very far removed from the idea of computing normal forms and similar concepts that we have described in the beginning of \cref{subsec:intro-1}.
In the type-theoretic setting, the role of the reduction relation is played by a type family $R : A \times A \to \UU$, and the statement that $a:A$ reduces to $b:A$ simply becomes the type $R(a,b)$.
It is standard to give a relation a name such as $\sim$ which is then used infix, i.e.\ one writes $a \sim b$ or $a \torel b$.
In our type-theoretic applications, it will in general be undecidable whether a given element $a$ can be reduced, whether $a$ reduces to $b$, and even whether $a$ and $b$ are path-equal.
In particular, it is not possible to compute normal forms.


\subsection{Set-quotients and the usefulness of a homotopy basis} \label{subsec:intro-hott-application}

Let us try to describe fairly concretely why a homotopy basis is useful in homotopy type theory.
We assume basic familiarity with the contents and notations of the book~\cite{hott-book}, the terminology of which we use.
Recall that a type $A$ is a \emph{proposition} if any two of its elements are equal,
\begin{equation}
\mathsf{isProp}(A) \defeq \Pi(x \, y : A). x = y,
\end{equation}
and a \emph{set} if any two parallel equalities are equal, i.e.\ if every equality type $x = y$ is a proposition,
\begin{equation}\label{eq:isSet}
\mathsf{isSet}(A) \defeq \Pi(x \, y : A). \mathsf{isProp}(x = y).
\end{equation}
One also says that sets are the types that satisfy the principle of unique identity proofs (UIP).

Recall from \cite[Chp 6.10]{hott-book} that, for a given type $A : \UU$ together with a relation
$(\sim) : A \to A \to \UU$, the set-quotient can be implemented as the higher inductive type
\begin{equation} \label{eq:setquotient}
\begin{aligned}
& \text{inductive } \setquotsim \text{ where} \\
&\qquad \iota : A \to \setquotsim \\
&\qquad \glue : \Pi\{a,b : A\}. (a \sim b) \to \iota(a) = \iota(b) \\
&\qquad \mathsf{trunc} : \Pi\{x,y : \setquotsim\}. \Pi(p,q : x=y). p=q
\end{aligned}
\end{equation}
The last constructor $\mathsf{trunc}$ ensures that the type $\setquotsim$ is a set.
From the above representation, we can derive the usual elimination rule for the set-quotient:
In order to get a function $f : ({\setquotsim}) \to X$, we need to give a function $g : A \to X$ such that, whenever $a \sim b$, we have $g(a) = g(b)$.
However, this only works if $X$ is a set itself.
If it is not, we have a priori no way of constructing the function $f$.

Let us look at one instance of the problem. We consider the following set-quotient, which we will use as a running example.
It is a standard construction that has been discussed in \cite[Chp 6.11]{hott-book}.
\begin{example}[free group] \label{ex:fg}
	Let $M$ be a set.
	We construct the free group on $M$ as a set-quotient.
	We consider lists over $M \uplus M$, where we think of the left copy of $M$ as positive and the right copy as negative elements.
	For $x : M \uplus M$, we write $x^{-1}$ for the ``inverted'' element: 
	\begin{equation}
	\inl(a)^{-1} \defeq \inr(a) \qquad \qquad \qquad \inr(a)^{-1} \defeq \inl(a)
	\end{equation}
	We let the binary relation $\torel$ on $\List(M \uplus M)$ to be generated by the reduction steps
	for all lists $[\ldots, x_i, \ldots]$:
	\begin{alignat}{3} \label{eq:fg-relation}
	& [\ldots, x_1, x_2, x_2^{-1}, x_3, \ldots] & \; &\torel & \;\; & [\ldots, x_1, x_3, \ldots]. 
	\end{alignat}
	Then, the set-quotient $\setquotXY{\List(M \uplus M)}{\mathop{\torel}}$ is the free group on $M$: It satisfies the correct universal property by \cite[Thm 6.11.7]{hott-book}.
\end{example}

Another way to construct the free group on $M$ is to re-use the natural groupoid structure that every type carries; this can be seen as a typical ``homotopy type theory style'' construction.
It works as follows.
The \emph{wedge of $M$-many circles} is the (homotopy) coequaliser of two copies of the map $M$ into the unit type, 
$\hcolim (M \rightrightarrows \unit)$.
Using a higher inductive type, it can be explicitly constructed:
\begin{equation} \label{wedge}
\begin{aligned}
& \text{inductive } \hcolim (M \rightrightarrows \unit) : \UU \text{ where} \\
& \qquad \mathsf{base}: \hcolim (M \rightrightarrows \unit) \\
& \qquad \mathsf{loop}: M \to \mathsf{base} = \mathsf{base}
\end{aligned}
\end{equation}
Its loop space $\Omega(\hcolim (M \rightrightarrows \unit))$ is by definition simply $\mathsf{base} = \mathsf{base}$.
This loop space carries the structure of a group in the obvious way: the neutral element is given by reflexivity, multiplication is given by path composition, symmetry by path reversal, and every $a:M$ gives rise to a group element $\mathsf{loop}(a)$.
This construction is works without the assumption that $M$ is a set and defines the free \emph{higher} group (cf.~\cite{Kraus:free}).
Following \cite{bezem_symmetry}, we write $\fhg(M)$ for this free higher group:
\begin{equation} \label{eq:FA-definition}
\fhg(M) \defeq  \Omega(\hcolim (M \rightrightarrows \unit)) \text{.}
\end{equation}
In contrast to this observation, the set-quotient of \cref{ex:fg} 
ignores any existing higher structure (cf. \cite[Rem 6.11.8]{hott-book}) and thus really only defines the free ``ordinary'' (set-level) group.
If we do start with a set $M$, it is a natural question whether the free higher group and the free group coincide:
There is a canonical function 
\begin{equation}\label{eq:higher-to-ordinary-fg}
\fhg(M) \to \quotX{\List(M \uplus M)},
\end{equation}
defined analogously to $\Omega(\mathsf{S}^1) \to \mathbb Z$, cf.\ \cite{hott-book}.
Classically, this function is an equivalence.
Constructively, it is an open problem to construct an inverse of \eqref{eq:higher-to-ordinary-fg}.

The difficulties do not stem from the first two constructors of the set-quotient.
Indeed, we have a canonical map
\begin{equation} \label{eq:omega1}
\omega_1 : \List(M \uplus M) \to \fhg(M)
\end{equation}
which maps a list such as $[\inl(a_1), \inr(a_2), \inl(a_3)]$ to the composition of paths given as $\mathsf{loop}(a_1) \ct (\mathsf{loop}(a_2))^{-1} \ct \mathsf{loop}(a_3)$.
For this map, we also have
\begin{equation} \label{eq:omega2}
\omega_2 : \Pi(\ell_1,\ell_2 : \List(M \uplus M)). (\ell_1 \torel \ell_2) \to \omega_1(\ell_1) = \omega_1(\ell_2)
\end{equation}
since consecutive inverse loops cancel each other out.
Therefore, if we define $(\quot)$ to be the higher inductive type \eqref{eq:setquotient} \emph{without} the constructor $\mathsf{trunc}$, i.e.\ the \emph{untruncated quotient} or \emph{coequaliser}, then there is a canonical map
\begin{equation} \label{eq:omega-complete}
\omega : \quotX{\List(M \uplus M)} \to \fhg(M).
\end{equation}
Thus, the difficulty with defining an inverse of \eqref{eq:higher-to-ordinary-fg} lies solely in the question whether $\fhg(M)$ is a set.
This is an open problem which has frequently been discussed in the homotopy type theory community (a slight variation is recorded in \cite[Ex 8.2]{hott-book}).
It is well-known in the community how to circumvent the problem if $M$ has decidable equality.
However, the only piece of progress on the general question that we are aware of is the result in \cite{Kraus:free}, where it is shown that all fundamental groups \cite[Chp 6.11]{hott-book} are trivial.
In other words: Instead of showing that \emph{everything} above truncation level $0$ is trivial, the result shows that a single level is trivial.
The proof in \cite{Kraus:free} uses a rather intricate construction which is precisely tailored to the situation.

The construction of functions $({\setquot}) \to X$ motivates the connection to higher-dimensional rewriting.
We do not allow an arbitrary type $X$, however; instead, we assume that $X$ is $1$-truncated, i.e.\ a \emph{groupoid} a.k.a. a \emph{1-type}, which means that all path spaces of $X$ are sets,
\begin{equation}
\mathsf{isGrp}(A) \; \defeq \; \Pi(x\, y : A). \mathsf{isSet}(x = y).
\end{equation}
As an application, we will give a new proof for the theorem that the fundamental groups of $\fhg(M)$ are trivial.
We will also show a family of similar statements, by proving a common generalisation.

The characterisation of the equality types of $({\setquot})$ makes it necessary to consider \emph{closed zig-zags} in $A$.
A closed zig-zag is simply an element of the symmetric-reflexive-transitive closure, for example:

\newdimen\WF
\newdimen\WS

\WF=.475\textwidth
\WS=.475\textwidth
\begin{minipage}[t]{\WF}
		\begin{align*}
		& s : a \torel b \qquad p : d \torel c \\
		& t : c \torel b \qquad q : a \torel d
		\end{align*}
\end{minipage}%
\begin{minipage}[t]{\WS}
	\begin{equation} \label{eq:cyclepicture}
	\begin{tikzpicture}[x=1.5cm,y=-1.50cm,baseline=(current bounding box.center)]
	\tikzset{arrow/.style={shorten >=0.1cm,shorten <=.1cm,-latex}}
	\node (A) at (0,0) {$a$}; 
	\node (B) at (0,1) {$b$}; 
	\node (C) at (1,1) {$c$}; 
	\node (D) at (1,0) {$d$}; 
	
	\draw[arrow] (A) to node [left] {$s$} (B);
	\draw[arrow] (C) to node [below] {$t$} (B);
	\draw[arrow] (D) to node [right] {$p$} (C);
	\draw[arrow] (A) to node [above] {$q$} (D);
	\end{tikzpicture}
	\end{equation}
\end{minipage}

\noindent
Our first result related to quotients (\cref{thm:gensetquotnew}) says:
We get a function $({\setquot}) \to X$ into a groupoid $X$ if we have $f : A \to X$ and $h : (a \torel b) \to f(a) = f(b)$, together with the coherence condition stating that $h$ maps any closed zig-zag to a ``commuting cycle'' in $X$. In the case of the example \eqref{eq:cyclepicture} above, this means that the composition $h(s) \ct h(t)^{-1} \ct h(p)^{-1} \ct h(q)^{-1}$ equals $\refl_{f(a)}$.
\cref{thm:gensetquotnew} is fairly simple, and we do not consider it a major contribution of this paper.

The actual contribution of the paper is to make \cref{thm:gensetquotnew} usable
since, on its own, it is virtually impossible to apply in any non-trivial situation.
The reason for this is that the coherence condition talks about \emph{closed} zig-zags.
Zig-zags are inductively generated (they are simply a chain of segments), but closed zig-zags are not.
If we have a property of closed zig-zags which we cannot generalise to arbitrary zig-zags, then there is no obvious inductive strategy to show the property in general: if we remove a segment of a closed zig-zag, it is not closed any more.
In all our examples, it seems not possible to formulate an induction hypothesis based on isolated not-necessarily-closed zig-zags.%
\footnote{To avoid the problem that closedness poses here, one may try to instead consider \emph{two} parallel but not-necessarily-closed zig-zags. However, the type of such pairs is not inductively generated either, and the situation is essentially the same. This is reminiscent of the well-known fact in type theory that the principles \texttt{UIP} (parallel equalities are equal) and \texttt{Axiom K} (loops are equal to reflexivity) are interderivable.}

Thus, how can \cref{thm:gensetquotnew} be made usable?
This is where the construction of a homotopy basis comes into play.
In the example of the free group, the relation $\torel$ on $\List{(M \uplus M)}$ can be presented in a way that ensures that it is Noetherian and confluent,
conditions that are (stronger than) necessary in order to construct the homotopy basis, and we have full control over how we want this basis to look like.
While it is very hard to show a property directly for \emph{all} closed zig-zags, it is much more manageable to show the property for closed zig-zags in the homotopy basis, and if the property is nice enough (which it is in all our examples),
then this is sufficient.

Let us get back to homotopy type theory.
The combination of the two mentioned results (\cref{thm:gensetquotnew} and \cref{thm:unary-basis}) gives us \cref{thm:dirsetquot}: Given a groupoid $X$ and $f : A \to X$ such that $a \torel b$ implies $f(a) = f(b)$, it suffices to show that closed zig-zags in the basis are mapped to trivial equality proofs.
We apply this to show that the free higher group over a set has trivial fundamental groups.
There is a family of similar statements that we also discuss and prove.

\subsection{Structure of this paper and formalisation} \label{subsec:structure-and-agda}

The ideas and developments in this paper can be split into three categories:

\begin{enumerate}
	\item \label{item:cat-1}
	The first category consists of the plain rewriting arguments that are to some degree independent of the foundation in which they are formulated. These arguments work in set-theoretic settings as well as in various forms of type theory.
	\item \label{item:cat-2}
	The second category addresses the choices that have to be made when choosing homotopy type theory as the foundational setting.
	A central question is whether one works with general types or \emph{h-sets}, i.e.\ types of truncation level 0.
	\item \label{item:cat-3}
	The final category consists of our concrete applications in homotopy type theory.
\end{enumerate}
%
We strive to separate the rewriting arguments from the type-theoretic arguments as much as possible and structure the development as follows:

\begin{itemize}
	\item In \cref{sec:non-tt}, we show the specific construction of a homotopy basis in a standard generic (unspecified) set-theoretic framework.
	This corresponds to category \ref{item:cat-1} above.
	We attempt to follow the style and presentation of other papers in the field on higher dimensional rewriting.
	\item The fairly short \cref{subsec:translation-is-generalisation} points out which parts of the 
	development presented in \cref{sec:non-tt} require particular attention when switching to homotopy type theory, and in which ways our type-theoretic formulation is more general. In this section, we give a high-level explanation of these points that corresponds to category \ref{item:cat-2}.
	\item \cref{sec:tt} gives the complete type-theoretic translation.
	\item We have formalised the development of \cref{sec:tt} in Agda. The source code is available at \href{https://bitbucket.org/fplab/confluencecoherence}{\nolinkurl{bitbucket.org/fplab/confluencecoherence}} and requires Agda 2.6.2.2 to be installed.
	In addition, we make a browsable \texttt{html} version available at \href{http://www.cs.nott.ac.uk/~psznk/agda/confluence/}{\nolinkurl{cs.nott.ac.uk/~psznk/agda/confluence/}}.
	This \texttt{html} version requires no software to be installed, is fully interlinked (i.e.\ everything can be clicked on to reach its definition), and benefits from complete syntax highlighting.
	\item Finally, we explain the applications in homotopy type theory in
	\cref{sec:applications} (category \ref{item:cat-3}).
\end{itemize}

%
\begin{wrapfigure}{R}{0.5\textwidth}
\centering
\begin{tikzpicture}[baseline=(current bounding box.center), x=2cm, y=-.1cm]
\tikzset{arrow/.style={shorten >=.1cm,shorten <=.1cm,-{>[length=1.3mm]}}}
\tikzset{node/.style={rectangle,draw}}
\node[rectangle,draw] (S1) at (1, -10) {\cref{sec:intro}};
\node[rectangle,draw] (S2) at (0, 0) {\cref{sec:non-tt}};
\node[rectangle,draw] (S3) at (0, 10) {\cref{subsec:translation-is-generalisation}};
\node[rectangle,draw] (S4) at (1, 5) {\cref{sec:tt}};
\node[rectangle,draw] (Agda) at (2, 5) {\href{http://www.cs.nott.ac.uk/~psznk/agda/confluence/}{Agda}};
\node[rectangle,draw] (S5) at (1, 20) {\cref{sec:applications}};
\draw[arrow] (S1.south west) to node [] {} (S2.north);
\draw[arrow] (S1.south) to node [] {} (S4.north);
\draw[arrow] (S1.south east) to node [] {} (Agda.north);
\draw[arrow] (S2) to node [] {} (S3);
\draw[arrow] (S3.south) to node [] {} (S5.north west);
\draw[arrow] (S4.south) to node [] {} (S5.north);
\draw[arrow] (Agda.south) to node [] {} (S5.north east);
\draw[-, shorten >=.1cm,shorten <=.1cm] (Agda.west) to node [] {} (S4.east);
\end{tikzpicture}
\end{wrapfigure}

\noindent
This structure allows a reader who is interested in the rewriting argument, but less so in type theory, to only study \cref{sec:non-tt}.
At the same time, an expert in the field of homotopy type theory will probably understand all ideas by additionally reading \cref{subsec:translation-is-generalisation} without going through
\cref{sec:tt} in full.

However, a reader interested in the full type-theoretic development and all the nitty-gritty details may wish to skip \cref{sec:non-tt,subsec:translation-is-generalisation} completely and immediately jump to \cref{sec:tt} as well as the accompanying Agda formalisation: \cref{sec:tt} can be seen as a high-level guide through the Agda code, and it comes with links to the \texttt{html} version of all the important definitions and theorems.

Finally, \cref{sec:applications} explains our applications in homotopy type theory.

\subsection{Background of the paper}
The core observation on which this article is based is that confluence and wellfoundedness can be used to prove coherence results internally in homotopy type theory.
We (the current authors) originally presented this insight at the LICS'20 conference~\cite{krausVonRaumer:wellfounded}.
We received very helpful feedback.
In particular, 
Vincent van Oostrom explained to us the connection to several lines of work in the rewriting community, especially the relationship to Squier's work.
The current article extends and improves the conference paper~\cite{krausVonRaumer:wellfounded}, and attempts to more cleanly separate the rewriting arguments from the type-theoretic applications.

The arguments of the current article are somewhat different and the results in a certain way more general than the results in the conference paper.
By specialising the main result of \cref{sec:tt}, i.e.\ the construction of a homotopy basis in type theory (\cref{thm:tt-basis}), we can derive a version of the main result of the conference paper (\cref{thm:unary-basis}); this then turns out to actually be slightly weaker in a subtle sense, as explained in \cref{rem:lics-paper-vs-current-article}.
However, these differences are insignificant from the point of view of the applications that we present.

The conference paper was presented with a formalisation in the Lean theorem prover, available at
\href{https://gitlab.com/fplab/freealgstr}{\nolinkurl{gitlab.com/fplab/freealgstr}}. This formalisation is independent of, and follows a different strategy from, our Agda implementation.


\section{A homotopy basis for Noetherian 2-polygraphs}\label{sec:non-tt}

In this section, we explain the construction of a homotopy basis 
in a generic set-theoretic framework.
In later sections, we will show how the development can be translated to and generalised in homotopy type theory, and used to address open questions in the field.

\subsection{Notations for 1-polygraphs}

A \emph{1-polygraph} is given by two sets $\Sigma_0$ and $\Sigma_1$ together with two functions $s_0, t_0 : \Sigma_1 \to \Sigma_0$.
This structure is sometimes known as a a \emph{quiver} \cite{gabriel1972unzerlegbare}, a \emph{directed pseudograph}, or simply a \emph{directed graph}.
Alternatively, it may be described as a low-dimensional special case of Street's \emph{computads} \cite{street1987algebra}, analogously to how it is a special case of Burroni's \emph{polygraphs} \cite{burroni1993higher}.

We refer to elements of $\Sigma_0$ as \emph{objects} and elements of $\Sigma_1$ as \emph{reduction steps} or simply \emph{steps}.
Given $u \in \Sigma_1$, we call $s_0(u)$ the \emph{source} and $t_0(u)$ the \emph{target} of $u$.
For $x,y \in \Sigma_0$, we write $(x \torel y)$ for the subset of $\Sigma_1$ containing those reduction steps which have $x$ as source and $y$ as target.

We write $(x \tostar y)$ for the set of composable (if $x = y$ possibly empty) finite sequences of reduction steps which start in $x$ and end in $y$.
Elements of $(x \tostar y)$ are \emph{reduction sequences} or simply \emph{sequences} from $x$ to $y$.
Formally, such a sequence consists of an \emph{object sequence} $(x = x_0, x_1, \ldots, x_n = y)$, with $x_i \in \Sigma_0$, and a list $(u_1, \ldots, u_n)$ with $u_i \in (x_{i-1} \torel x_i)$.
Given two composable sequences, say $u \in (x \tostar y)$ and $v \in (y \tostar z)$, we write $u \cdot v$ for their composition, $u \cdot v \in (x \tostar z)$. We use the same notation if $u$ and/or $v$ is a single step instead of a sequence.

$(y \from x)$ denotes a copy of the set $(x \torel y)$.
We write $(x \tofrom y)$ for the disjoint sum $(x \torel y) \uplus (x \from y)$
and $(x \tofromstar y)$ for the set of \emph{reduction zig-zags} (or simply \emph{zig-zags}) from $x$ to $y$.
Just as a reduction sequence, a reduction zig-zag is given by an object sequences and a list of steps
$(u_1, \ldots, u_n)$, where however we allow that either $u_i \in (x_{i-1} \torel x_i)$ or $u_i \in (x_{i-1} \from x_i)$.
If for all steps the first (second) option is the case, the zig-zag is called \emph{positive} (\emph{negative}).
Based on this, we use notations such as $(y \from x \torel z)$, which is the set of pairs $(u,v)$ with $u \in (x \torel y)$ and $v \in (x \torel z)$.
In this case, we write $u^{-1} \cdot v \in (y \from x \torel z)$ to make clear that $u$ has been formally inverted, and $u^{-1} \cdot v$ is seen as a reduction zig-zag.

Finally, we denote by $\Sigma_1^*$ the union of all sets of the form $(x \tostar y)$, i.e.\ the set of all sequences.
Note that we have one trivial sequence of length $0$ for each $x \in \Sigma_0$, which we denote by $\varepsilon_x$.
Similarly, $(\Sigma_1 \uplus \Sigma_1^{-1})^*$ denotes the union of all sets of the form $(x \tofromstar y)$, i.e.\ the set of all zig-zags.
For $u \in (\Sigma_1 \uplus \Sigma_1^{-1})^*$, we write $s_0(u)$ for the starting element of the object sequence (i.e.\ $x_0$ in the description above) and $t_0(u)$ for the last (i.e.\ $x_n$).

\subsection{Terminating 1-polygraphs}

Let $>$ be a relation on a set $M$.
We call an element $n \in M$ \emph{accessible} if all $m \in M$ with $n > m$ are accessible.
Recall that $>$ is called \emph{Noetherian} (or \emph{co-wellfounded}) if all elements of $M$ are accessible.
Recall further that, for such a relation, we can perform \emph{Noetherian induction}: Given a property on $M$, if the property holds for an $n \in M$ as soon as it holds for all $m$ with $n > m$, then the property holds for all $n$.
Note that a relation is Noetherian if and only if its transitive closure is.


Let a 1-polygraph $\Sigma_0 \leftleftarrows \Sigma_1$ be given.
We say that the polygraph is equipped with a (Noetherian) order if we have a (Noetherian) order $>$ on the set $\Sigma_0$. 
%
We say that the polygraph is \emph{terminating} if it is equipped with a transitive Noetherian order $>$ and,
whenever there is a $u \in (x \torel y)$, we have $x > y$.
We furthermore extend an ordering to zig-zags by e.\,g. saying that for
$u \in \Sigma_1$ we have $x > u$ if and only if $x > y_i$ for all $y_i \in \Sigma_0$ in the
object sequence of $u$.


\begin{remark}
	Note that the above definitions of \emph{accessibility} and \emph{Noetherian}
	are phrased in a way that makes them usable in a constructive setting.
	It is a consequence that there exists no infinite sequence $x_0 > x_1 > \ldots$; however, that statement taken as a definition would not allow us to perform the constructions we do in this paper.
\end{remark}
\begin{remark}
	Accessibility could be formulated directly in terms of $\Sigma_1$ instead of referring to a relation $>$.
	The reason for introducing $>$ is that requiring $x > y$ can be a weaker condition than requiring $x \leadstostar y$, cf.~\cref{rem:>-could-be-truncation}.
\end{remark}

\subsection{The list extension of a Noetherian relation}
\label{subsec:list-extension}

Given a set $M$, we write $M^*$ for the set of finite (possibly empty) lists.
Given a relation $>$ on $M$, we extend it to a relation $\listext >$ on $M^*$, mirroring
the \emph{multiset extension} by Dershowitz and Manna~\cite{derschowitz-manna:multiset}.
This \emph{list extension} is defined as follows, where we choose to build in the transitive closure.
For lists $\vec n, \vec m \in M^*$, we have $\vec n \listext > \vec m$ if it is possible to transform $\vec n$ into $\vec m$ by applying the following operation one or multiple times:
remove one element $n$ of the list and replace it by a finite list, where each new list element has to be smaller than $n$.
In other words, the list extension of $>$ is the smallest relation $\listext >$ on $M^*$ which is:
\begin{enumerate}
	\item transitive;
	\item closed under congruence, 
	i.e.\ if $\vec k, \vec l, \vec m, \vec n \in M^*$ are four lists such that $\vec n > \vec m$, then we also have $(\vec k \cdot \vec n \cdot \vec l) > (\vec k \cdot \vec m \cdot \vec l)$;
	\item and, if $n \in M$ and $(n_1, \ldots, n_e) \in M^*$ such that $\forall i \in \{1, \ldots, e\}. n > n_i$, then $(n) > (n_1, \ldots, n_e)$.
	Here, $(n)$ is the list of length 1 with the single element $n$ and $\cdot$ donates list concatenation.
\end{enumerate}
Recall that the multiset extension of a Noetherian relation is Noetherian \cite{derschowitz-manna:multiset}.
As the multiset extension subsumes the list extension, the same statement holds for the list extension.
The proof that we give is an adaption of an argument by Nipkov \cite{nipkow:multiset}.
\begin{lemma} \label{lem:list-ext-noetherian}
	Let $M$ be a set with a relation $>$.
	If $>$ is Noetherian on $M$, then the list extension $\listext >$ 
	on the set $M^*$ is also Noetherian.
\end{lemma}
\begin{proof}
	We need to show that every list in $M^*$ is accessible.
	Since accessibility is closed under transitive closure, we can without
	loss of generality assume the
	definition of $>_L$ to lack the transitive closure and only contain the
	other two closure properties.
	We do this in three steps:
	\begin{enumerate}[label=\textbf{(\arabic*)},ref={(\arabic*)}]
		\item \label{enum:comp-acc} If lists $\ell_1$ and $\ell_2$ are accessible,
		then so is their concatenation $\ell_1 \cdot \ell_2$.
		\item \label{enum:singleton-acc} Single element lists $(m)$, with $m\in M$, are accessible.
		\item \label{enum:all-acc} Arbitrary lists of objects are accessible.
	\end{enumerate}
	Point \ref{enum:comp-acc} holds because, in order to make arrive at a list smaller than $\ell_1 \cdot \ell_2$, one has to make $\ell_1$ smaller or $\ell_2$ smaller. More precisely, any list smaller than $\ell_1 \cdot \ell_2$ is of the form $\ell_1' \cdot \ell_2$ or of the form $\ell_1 \cdot \ell_2'$ or of the form $\ell_1' \cdot \ell_2'$, with $\ell_1 >_L \ell_1'$ and $\ell_2 >_L \ell_2'$. Since $\ell_1$ and $\ell_2$ are individually accessible, this shows that their concatenation is.%
	\footnote{The more general statement is \emph{nested wellfounded induction}, discussed in the type-theoretic setting below, cf.\ \Cref{lem:nested-ind}.} 
	
	To prove \ref{enum:singleton-acc}, we apply Noetherian induction on $x \in M$.
	To show that $(x)$ is accessible, assume $(x) >_L \ell$ with $\ell = (x_1, \ldots x_n)$. We have to show that $\ell$ is accessible.
	By the induction hypothesis, each $(x_i)$ is accessible. Further, we have $\ell = (x_1) \cdot \ldots \cdot (x_n)$. Thus, the statement is given by point \ref{enum:comp-acc}.
	
	For the proof of \ref{enum:all-acc}, we now combine the two previous steps and use
	the same argument as before:
	Every list can be written as a concatenation of singleton lists and is therefore
	accessible.
\end{proof}

\subsection{Generalised 2-polygraphs} \label{subsec:gen-2-poly}

Burroni's notion of a \emph{2-polygraph} \cite{burroni1993higher}
extends a 1-polygraph $\Sigma_0 \leftleftarrows \Sigma_1$.
The extension consists of a set $\Sigma_2$ together with two functions $s_1, t_1 : \Sigma_2 \to \Sigma_1^*$ (i.e.\ a second 1-polygraph $\Sigma_1^* \leftleftarrows \Sigma_2$) subject to the condition that, for each $\alpha \in \Sigma_2$,
we have $s_0(s_1(\alpha)) = s_0(t_1(\alpha))$ and $t_0(s_1(\alpha)) = t_0(t_1(\alpha))$.
The data that we want to work with is captured by a slight generalisation of a 2-polygraph that is obtained by replacing $\Sigma_1^*$ by $(\Sigma_1 \uplus \Sigma_1^{-1})^*$, i.e.\ we generalise sequences to zig-zags.
Given a generalised 2-polygraph with two zig-zags $u, v \in (x \tofromstar y)$, we write $(u \To v)$ for the set of all $\alpha \in \Sigma_2$ with $s_1(\alpha) = u$ and $t_1(\alpha) = v$ and call $\alpha$ a \emph{rewrite step} from $u$ to $v$.
We use the notations $(u \Tostar y)$ as well as $(u \Tofrom v)$ and $(u \Tofromstar v)$ analogously to $(x \tostar y)$, $(x \tofrom y)$, and $(x \tofromstar y)$, respectively.

We say that a generalised 2-polygraph is terminating if the underlying 1-polygraph $\Sigma_0 \leftleftarrows \Sigma_1$ is terminating.
Further, we say that a generalised 2-polygraph is \emph{closed under congruence} if, for any $\alpha \in (\Sigma_2 \uplus \Sigma_2^{-1})^*$ and $u, v \in (\Sigma_1 \uplus \Sigma_1^{-1})^*$
with $s_0(s_1(\alpha)) = t_0(u)$ and $t_0(t_1(\alpha)) = s_0(v)$, we have a chosen
zig-zag of 2-cells
$u \cdot \alpha \cdot v \in (\Sigma_2 \uplus \Sigma_2^{-1})^*$ with
$s_1 (u \cdot \alpha \cdot v) = u \cdot s_1(\alpha) \cdot v$
and
$t_1 (u \cdot \alpha \cdot v) = u \cdot t_1(\alpha) \cdot v$.
\footnote{In the terminology of bicategories, the expression $u \cdot \alpha \cdot v$ corresponds to the \emph{horizontal composition} of the identities on $u$ and $v$ with $\alpha$. We do not require \emph{vertical composition}, which would correspond to an operation which turns a rewrite sequence in $(u \Tostar v)$ into a step in $(u \To v)$.}

\begin{remark}
	The fact that a \emph{rewrite step} in our setting is simply an element of $\Sigma_2$ is in contrast with the terminology of other authors (e.g.\ Alleaume and Malbos \cite{alleaume2016coherence}), for whom a rewrite step is a compositions of the form $u \cdot \alpha \cdot v$.
	For polygraphs that are closed under congruence, this distinction becomes essentially irrelevant since we are mostly interested in $(\Sigma_2 \uplus \Sigma_2^{-1})^*$ rather than $\Sigma_2$.
\end{remark}

\begin{remark} \label{rem:fibration-vs-families}
	Note that the data $(\Sigma_0, \Sigma_1, \Sigma_2, s_0, t_0, s_1, t_1)$ of a generalised 2-polygraph can equivalently be described as a 1-polygraph $\Sigma_0 \leftleftarrows \Sigma_1$ and, for each pair $x,y \in \Sigma_0$,
	another 1-polygraph $(x \tofromstar y) \leftleftarrows \Sigma_2^{x,y}$.
	Unfolding further, this data consists of: a set $\Sigma_0$; for each pair $x, y \in \Sigma_0$, a set $(x \torel y)$ of reduction steps; and, for each pair $x,y \in \Sigma_0$ and $u,v \in (x \tofromstar y)$, a set $(u \To v)$ of rewrite steps.
	This presentation is much more natural in type theory and will be used in \cref{sec:tt}.
\end{remark}

\subsection{Generalisations of Newman's Lemma}
Various notions of \emph{confluence} have been studied to characterize
well-behaved rewriting systems.
In this section, we will recall four of these definitions, adapt them to our
constructive meta-theory, and compare them with each other.
The notions that we consider are shown in \cref{fig:confluences}.


We will use the following terminology to refer to the global and local 
``topography'' of reduction sequences:
We call $u \in (y \fromstar x \tostar z)$ a \emph{peak} and $v \in (y \tostar x \fromstar z)$
a \emph{valley}.
If we replace the respective reduction sequences by single reduction steps, $u \in (y \from x \torel z)$ and $v \in (y \torel x \from z)$
will be called a \emph{local peak} (or \emph{span}) and \emph{local valley} (or \emph{co-span}),
respectively.

The first property we want to take a glance at is the arguably weakest notion
of confluence we want to consider. It requires that each local peak can be rewritten
into a reduction zig-zag below that peak.
It was originally formulated as a property of a rewriting system by Winkler and Buchberger \cite{buchberger1983criterion}.

\begin{definition} \label{def:WBstructure}
A \emph{Winkler-Buchberger structure} on a 2-polygraph with an order $>$ on its objects
consists of,
for each local peak $u \in (y \from x \torel z)$,
a rewrite zig-zag $\WB(u) \in (\Sigma_2 \uplus \Sigma_2^{-1})^*$ such that
$s_1(\WB(u)) = u$ and, writing $t_1(\WB(u)) = (y \tofrom x_1 \ldots x_n \tofrom z)$, we have $x > x_i$ for all $1 \leq i \leq n$.
\end{definition}

Just like we represent the Winkler-Buchberger property by a structure on the polygraph,
we will proceed to present (local) confluence as a choice of rewrite zig-zags:

\begin{definition}\label{def:LCstructure}
A \emph{local confluence structure} on the assumed 2-polygraph consists of, 
for each local peak $u \in (y \from x \torel z)$,
a rewrite zig-zag $\LC(u) \in (\Sigma_2 \uplus \Sigma_2^{-1})^*$ such that
$s_1(\LC(u)) = u$ and $t_1(\LC(u))$ takes the form $(y \tostar x' \fromstar z)$
for some $x' \in \Sigma_0$ which we will call the \emph{reduct} of the local peak
$u$.

A \emph{confluence structure} is defined analogously, with a rewrite zig-zag $\CS(u)$ for an arbitrary (not necessarily local) peak $u \in (y \fromstar x \tostar z)$. 
\end{definition}

While sometimes conflated with confluence, we will call \emph{Church-Rosser structure}
the generalisation of confluence where objects are connected by an arbitrary reduction zig-zag.

\begin{definition}\label{def:CRstructure}
A choice of rewrite zig-zag $\CR(u)$ for every $u \in (y \tofromstar z)$ is called
a \emph{Church-Rosser structure} if $s_1(\CR(u)) = u$ and $t_1(\CR(u)) \in (y \tostar x' \fromstar z)$, for some $x' \in \Sigma_0$.
Like with confluence, we will call $x'$ the reduct of $u$.
\end{definition}

A Church-Rosser structure always contains a confluence structure (since every peak is a zig-zag). Similarly, a confluence structure contains a local confluence structure.
If $u \in (x \torel y)$ implies $f(x) > f(y)$ (and $>$ is transitive), which happens in particular if the generalised 2-polygraph is terminating,
then every local confluence structure trivially is a Winkler-Buchberger
structure as well.

We will now show that, if a 2-polygraph is terminating and closed under congruence (see \cref{subsec:gen-2-poly}), then these implications can be reversed.
In its core, the proof is simply the standard argument for Newman's Lemma \cite{Huet1980ConfluentRA}.

\begin{lemma}\label{lem:wb-to-cr}
Assume we have a generalised 2-polygraph that is closed under congruence and is terminating, with underlying transitive relation $>$.
Then, a Winkler-Buchberger structure $\WB$ allows us to construct a Church-Rosser structure $\CR$. 
\end{lemma}
\begin{proof}
Let $y, z \in \Sigma_0$ be given. 
The goal is to construct, for any reduction zig-zag $u \in (y \tofromstar z)$, a sequence $\CR(u) \in (\Sigma_2 \uplus \Sigma_2^{-1})^*$ such that $s_1(\CR(u)) = u$ and such that $t_1(\CR(u))$ is a valley.

The list extension of $>$ gives an order $\listext >$ on $\Sigma_0^*$
and, by taking the underlying object sequence of a zig-zag, this induces a relation on $(y \tofromstar z)$ which, by \cref{lem:list-ext-noetherian}, is Noetherian.
We show the goal by Noetherian induction on this relation.

A given $u \in (y \tofromstar z)$ is either of the form $(y \tostar x' \fromstar z)$, in which case we are done ($\CR(u)$ is the empty sequence),
or it contains a local peak and can be written as $u = v \cdot u' \cdot w$ with $v \in (y \tofromstar y')$, $u' \in (y' \from x \to z')$, $w \in (z' \tofromstar z)$.
The Winkler-Buchberger structure and the closure under congruence gives us a rewrite step 
\begin{equation} \label{eq:wb-step}
\left(v \cdot \WB(u') \cdot w\right) \in \left((v \cdot u' \cdot w) \To (v \cdot t_1(\WB(u')) \cdot w)\right).
\end{equation}
By construction of the list extension and by the condition on the Winkler-Buchberger structure, 
the induction hypothesis lets us assume that we already have a suitable $\CR(v \cdot t_1(\WB(u')) \cdot w) \in (\Sigma_2 \uplus \Sigma_2^{-1})^*$.
Concatenating \eqref{eq:wb-step} with that rewrite zig-zag gives $\CR(u)$.
\end{proof}

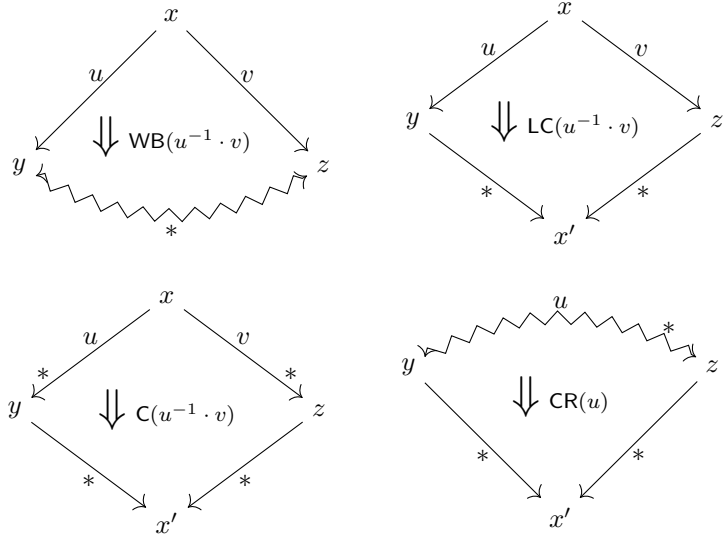
\begin{figure}
\centering
\begin{minipage}[t]{0.4\textwidth}
\centering
\begin{tikzpicture}[baseline=(current bounding box.center), x=2cm, y=2cm]
\tikzset{arrow/.style={-{>[length=1.3mm]}}}
\node (Px) at (0, 0) {$x$};
\node (Py) at (-1, -1) {$y$};
\node (Pz) at (1, -1) {$z$};
\node (PWB) at (0, -.8) {{\huge $\Downarrow$} {\footnotesize $\WB(u^{-1} \cdot v)$}};
\draw[arrow] (Px) to node [above] {$u$} (Py);
\draw[arrow] (Px) to node [above] {$v$} (Pz);
\draw[arrow, {<[length=1.3mm]}-{>[length=1.3mm]}, bend right,decorate,decoration=zigzag] (Py) to node [below] {$*$} (Pz);
\end{tikzpicture}
\end{minipage}
\begin{minipage}[t]{0.4\textwidth}
\centering
\begin{tikzpicture}[baseline=(current bounding box.center), x=2cm, y=1.5cm]
\tikzset{arrow/.style={-{>[length=1.3mm]}}}
\node (Px) at (0, 0) {$x$};
\node (Py) at (-1, -1) {$y$};
\node (Pz) at (1, -1) {$z$};
\node (Px') at (0, -2) {$x'$};
\node (PLC) at (0, -1) {{\huge $\Downarrow$} {\footnotesize $\LC(u^{-1} \cdot v)$}};
\draw[arrow] (Px) to node [above] {$u$} (Py);
\draw[arrow] (Px) to node [above] {$v$} (Pz);
\draw[arrow] (Py) to node [below] {$*$} (Px');
\draw[arrow] (Pz) to node [below] {$*$} (Px');
\end{tikzpicture}
\end{minipage} \\[1em]
\begin{minipage}[t]{0.4\textwidth}
\centering
\begin{tikzpicture}[baseline=(current bounding box.center), x=2cm, y=1.5cm]
\tikzset{arrow/.style={-{>[length=1.3mm]}}}
\node (Px) at (0, 0) {$x$};
\node (Py) at (-1, -1) {$y$};
\node (Pz) at (1, -1) {$z$};
\node (Px') at (0, -2) {$x'$};
\node (PLC) at (0, -1) {{\huge $\Downarrow$} {\footnotesize $\CS(u^{-1} \cdot v)$}};
\draw[arrow] (Px) to node [above] {$u$} node[pos=0.9,above] {$*$} (Py);
\draw[arrow] (Px) to node [above] {$v$} node[pos=0.9,above] {$*$} (Pz);
\draw[arrow] (Py) to node [below] {$*$} (Px');
\draw[arrow] (Pz) to node [below] {$*$} (Px');
\end{tikzpicture}
\end{minipage}
\begin{minipage}[t]{0.4\textwidth}
\centering
\begin{tikzpicture}[baseline=(current bounding box.center), x=2cm, y=2cm]
\tikzset{arrow/.style={-{>[length=1.3mm]}}}
\node (Py) at (-1, -1) {$y$};
\node (Pz) at (1, -1) {$z$};
\node (Px') at (0, -2) {$x'$};
\node (PLC) at (0, -1.2) {{\huge $\Downarrow$} {\footnotesize $\CR(u)$}};
\draw[arrow,{<[length=1.3mm]}-{>[length=1.3mm]},bend left,decorate,decoration=zigzag] (Py) to node [above]
  {$u$} node[pos=0.9,above] {$*$} (Pz);
\draw[arrow] (Py) to node [below] {$*$} (Px');
\draw[arrow] (Pz) to node [below] {$*$} (Px');
\end{tikzpicture}
\end{minipage}
\caption{Different notions of confluence.}
\label{fig:confluences}
\end{figure}

\subsection{A homotopy basis}

The goal of this section is to show that, with the help of a few assumptions, we can construct a \emph{homotopy basis} 
for a generalised 2-polygraph.

How does the concept of homotopy come into play here? Imagine a topological
realisation of the 2-polygraph where objects are represented by points,
reduction steps by the interval space, and rewrite steps by surfaces between the
zig-zags corresponding to their source and target.
Then, we can consider the fundamental group of this topological space.
The fundamental group is trivial (at any point) if,
for reduction zig-zags $u, v \in (x \tofromstar y)$,
there is always a ``filler'' $\alpha \in (u \Tofromstar v)$.
A subset of $\Sigma_2$ which can be used to ``fill'' every loop is a homotopy basis:


\begin{definition}[homotopy basis] \label{def:homotopy-basis}
	Let a generalised 2-polygraph $\Sigma = (\Sigma_0, \Sigma_1, \Sigma_2)$ be given.
	A \emph{homotopy basis} consists of, 
	for all $x,y \in \Sigma_0$ and
	$u,v \in (x \tofromstar y)$, a rewrite zig-zag
	$\alpha_{u,v} \in (u \Tofromstar v)$.
\end{definition}

\begin{remark}[alternative definition of homotopy basis]
	Instead of calling the set of all $\alpha_{u,v}$ a homotopy basis,
	it may seem more natural to refer to a subset $\mathcal B \subseteq \Sigma_2$ as a homotopy basis if one can construct all the $\alpha_{u,v}$ from rewrite steps in $\mathcal B$ (or their closure under congruence).
	That variation would in particular allow the formulation of properties or statements with respect to the cardinality of a basis.
	However, we are only interested in the collection of all the $\alpha_{u,v}$ itself, which is why we use the simplified \cref{def:homotopy-basis}.
\end{remark}

An additional but intuitive property that we need it the following:
\begin{definition}
	We say that a generalised 2-polygraph $(\Sigma_0, \Sigma_1, \Sigma_2)$ \emph{cancels inverses} if,
	for any reduction step $s \in (x \torel y)$, we have rewrite zig-zags $\RINV(s) \in (s \cdot s^{-1} \Tofromstar \varepsilon_x)$ and $\LINV(s) \in (s^{-1} \cdot s \Tofromstar \varepsilon_y)$. 
\end{definition}


\begin{theorem} \label{thm:main-theorem-settheory}
	Let $(\Sigma_0, \Sigma_1, \Sigma_2)$ be a terminating generalised 2-polygraph which is closed under congruence and which cancels inverses.
	If it has a Winkler-Buchberger structure $\WB$, then it has a homotopy basis.
\end{theorem}
\begin{proof}
By \cref{lem:wb-to-cr}, we can construct a Church-Rosser structure $\CR$.

Moreover, using $\RINV$ and $\LINV$ together with closure under congruence we can, by straightforward induction on the length of a zig-zag $u \in (y \tofromstar z)$,
construct a sequence $\INV(u) \in (u \cdot u^{-1} \To \varepsilon_y)$, where $\varepsilon_y$ is the empty zig-zag.
Given $u, v \in (y \tofromstar z)$ and a rewrite zig-zag $\alpha \in (u \cdot v^{-1} \Tofromstar \varepsilon_y)$, closure under congruence shows that we also have
\begin{alignat}{5}
& u &\quad & \Tofromstar &\quad & u \cdot v^{-1} \cdot v 	&\qquad\qquad & \text{by $\INV(v^{-1})$,} \\
&&&\Tofromstar && v && \text{by $\alpha$}
\end{alignat}

By the above argument, it suffices to show the goal for the case that $u$ is a \emph{closed} reduction zig-zag ($u \in (y \tofromstar y)$) and $v$ is empty ($v = \varepsilon_y$).
By Noetherian induction on the object $y$, we show the following statement:
\begin{equation*}
P(y) := \text{``For all closed reduction zig-zags $u \in (y \tofromstar y)$, we have $\alpha_u \in (u \Tofromstar \varepsilon_y)$.''}
\end{equation*}

To this end, we assume $P(z)$ for all $y > z$ and take an arbitrary
reduction zig-zag $u \in (y \tofromstar y)$ for which we want to construct
$\alpha_u \in (u \Tofromstar \varepsilon_y)$.
Now consider the rewrite zig-zag $\CR(u) \in (u \Tofromstar v \cdot w^{-1})$,
where $v, w \in (y \tostar z)$ for the reduct $z$ of $\CR(u)$.
If either of $v$ or $w$ is an empty sequence, then so is the other and we have $y = z$, in which case we set 
$\alpha_u := \CR(u)$.
Otherwise, we have $y > z$.
In this case we consider the closed reduction zig-zag $w^{-1} \cdot v \in (z \tofromstar z)$.

From the induction hypothesis, we obtain a rewrite zig-zag
$\alpha_{w^{-1} \cdot v} \in (w^{-1} \cdot v \Tofromstar \varepsilon_z)$.
We can now use this rewrite zig-zag to construct the following chain of rewrites, again heavily relying on the closure under congruence; 
see \Cref{fig:basis} for an illustration:
\begin{alignat}{5}
& u &\quad & \Tofromstar &\quad & v \cdot w^{-1}			&\qquad\qquad & \text{by $\CR(u)$,} \\
&&&\Tofromstar && v \cdot w^{-1} \cdot v \cdot v^{-1}	&& \text{by $\INV(v)$,} \\
&&&\Tofromstar && v \cdot v^{-1}			&& \text{by $v \cdot \alpha_{w^{-1} \cdot v} \cdot v^{-1}$,} \\
&&&\Tofromstar && \varepsilon_y				&& \text{by $\INV(v)$,}
\end{alignat}
which completes the construction.
\end{proof}

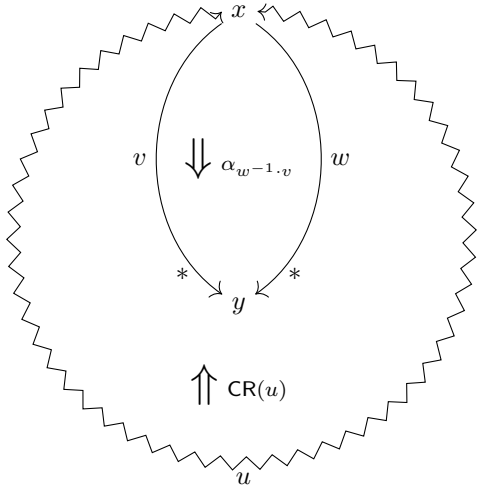
\begin{figure}[H]
\begin{tikzpicture}[baseline=(current bounding box.center), x=3cm, y=3cm]
\tikzset{arrow/.style={-{>[length=1.3mm]}}}
\node (Px) at (0, 1) {$x$};
\node (Py) at (0, -0.3) {$y$};
\node (PCR) at (0, -0.65) {{\huge $\Uparrow$} {\footnotesize $\CR(u)$}};
\node (Palpha) at (0, 0.35) {{\huge $\Downarrow$} {\footnotesize $\alpha_{w^{-1} \cdot v}$}};
\draw[arrow, {<[length=1.1mm]}-{>[length=1.1mm]},decorate,decoration=zigzag] (Px.west)arc(95:446:1) node [below, pos=0.5] {$u$};
\draw[arrow, bend right=55] (Px) to node [left] {$v$} node [left, pos=0.9] {$*$} (Py);
\draw[arrow, bend left=55] (Px) to node [right] {$w$} node [right, pos=0.9] {$*$} (Py);
\end{tikzpicture}
\caption{The induction step for the construction of the homotopy basis.}\label{fig:basis}
\end{figure}

\section{A translation to type theory: caveats and generalisations}
\label{subsec:translation-is-generalisation}

While the above constructions are formulated in an (unspecified) standard set-theoretic framework, the main motivation for the development
are applications in homotopy type theory.
Large parts of the translation of \cref{sec:non-tt} into homotopy type theory follow standard strategies and are to some degree mechanical, but some points require further attention.
Moreover, our type-theoretic construction is, due to the choices we make, more general than the set-theoretic development presented above. 
In this (short) section, we discuss these key points.
A reader working in homotopy type theory will likely find these explanations sufficient to understand how the type-theoretic formulation works and can then jump to \cref{sec:applications}, while someone interested in all details may wish to skip the current section 
and immediately go to \cref{sec:tt}.

\subsection{Constructivity}
When translating a development into
(constructive) type theory,
the natural first consideration is whether 
any possibly implicit use of
the law of excluded middle and its consequences, especially proof by double negation
and decidability of equality, can be avoided.
This is the case here but, as it is often happens, relies on formulating the definitions correctly;
an obvious example is the definition of \emph{wellfounded} (and \emph{Noetherian}), where the classical negative phrasing ``there is no infinite sequence'' would not allow a constructive argument, while 
a well-known inductive definition works (cf.\ \cref{sec:closures}).

\subsection{Coherence}
More specific to the setting of homotopy type theory is the phenomenon of higher equalities.
In many cases, translations of standard mathematical concepts into homotopy type theory are phrased using sets (i.e.\ types satisfying UIP, see equation \eqref{eq:isSet} on page \pageref{eq:isSet}) in order to faithfully represent the corresponding theory, and formulating the same concepts for arbitrary types can be extremely difficult or impossible.

We choose to not restrict ourselves to sets for the mere purpose of being more general (although the set-case would suffice for the applications in \cref{sec:applications}).
This is not particularly difficult, but a possibly surprising consequence is that there may be rewrite zig-zags $u : x \Tofromstar x$ of length zero which are \emph{not} the trivial sequence $\varepsilon_x$.
To be precise, the type of zig-zags of length zero from $x$ to $y$ is equivalent to the equality type $x = y$.
The analogue to \cref{thm:main-theorem-settheory} thus needs to include the assumption that the rewrite system \emph{cancels empty sequences}, a condition which becomes trivial for sets (cf.\ \cref{thm:dirsetquot}).

Similarly, the faithful translation of the relations of \cref{sec:non-tt} would be to consider families of \emph{propositions}, i.e.\ types with at most one inhabitant.
Again, we aim for greater generality and avoid this assumption, which however has little consequences for the overall argument.
There are several further points where the type-theoretic formulation is, strictly speaking, a generalisation of the results of \cref{sec:non-tt}.

\subsection{Formulation of results}
As we will discuss in \cref{subsec:noetherian-cycle-induction} below,
the construction of a homotopy basis is reminiscent of an induction principle,
and it is natural from a type-theoretic point of view to phrase it as such.
We will derive the following principle:

\begin{theorem*}[simplified formulation of \cref{thm:unary-basis}]
Let $A$ be a type with a binary relation $\leadsto$ that is Noetherian and locally confluent, and let $P$ be a type family indexed over closed zig-zags.
To prove (inhabit) $P$ for all closed zig-zags, it suffices to prove $P$ for empty zig-zags and for the diamonds that come from the local confluence property, and to check that $P$ is closed under standard groupoidal constructions.
\end{theorem*}

\section{A homotopy basis in homotopy type theory} 
\label{sec:tt}

This section serves as a high-level description of our Agda formalisation, which in turn presents the type-theoretic development in full detail.
The main mathematical ideas which make the arguments work correspond to those presented in \cref{sec:non-tt} and, keeping the caveats and remarks of \cref{subsec:translation-is-generalisation} in mind,
the type-theoretic formulation follows standard strategies.
The expert reader may wish to immediately jump to \cref{sec:applications}.

We use \cref{subsec:specify-type-theory} to specify the type theory that we work in.
For any such translation to type theory, the specific choices one makes determine whether the procedure is straightforward or includes (mathematical) challenges.
We explain in \cref{subsec:translation-is-generalisation} which choices we make and how they make the type-theoretic statement a generalisation of the results in \cref{sec:non-tt}.
The concrete step-by-step translation is split into two subsections:
In \cref{sec:closures}, we examine how closures of binary relations behave
in homotopy type theory,
and in \cref{sec:hott-polygraphs}, we discuss the construction of a homotopy basis of a generalised 2-polygraph.
We then show in \cref{subsec:noetherian-cycle-induction} how this work allows us to derive an induction-like statement similar to the one proved in our previous conference paper~\cite{krausVonRaumer:wellfounded}.

As explained in \cref{subsec:structure-and-agda} above, we have formalised this part of the paper in Agda
We link the main constructions and proofs to the \texttt{html} version of the formalisation.

\subsection{Homotopy type theory as the setting} \label{subsec:specify-type-theory}

The type theory we work in is homotopy type theory as developed in the book~\cite{hott-book}.
However, for the translation of the results of \cref{sec:non-tt} itself, we do not rely on any features specific to homotopy type theory:
All of what we do here (i.e.\ in \cref{sec:tt}) works in intensional Martin-L\"of type theory with function extensionality.
Only for the applications that we discuss in \cref{sec:applications},
we rely on having set-quotients and a univalent universe.

Regarding notation, we mostly follow the book~\cite{hott-book}.
For \emph{judgmental} (also known as \emph{definitional}) equality between expressions, we use the symbol $\equiv$, and for definitions, we write $\defeq$.

In detail, the type theory that we consider features the following components:
\begin{itemize}
\item We assume that the type theory has a (Russell-style) \emph{universe} $\UU$.
It is standard to assume a hierarchy 
\begin{equation*}
\UU_0 : \UU_1 : \UU_2 : \ldots
\end{equation*}
of universes, and in such a theory, $\UU$ may denote any universe $\UU_i$
(i.e.\ our constructions are \emph{universe polymorphic}). 
\item Besides types of non-dependent functions we require for any type 
$A : \UU$ and for any \emph{type family} $B : A \to \UU$ the
type $\Pi(a:A).B(a)$ of \emph{dependent functions}.
\item For any type $A : \UU$ with elements $a,b :A$, we have the 
Martin-L\"of \emph{equality type} which,
following \cite{hott-book}, we denote by $(a = b) : \UU$.
The equality type is generated inductively on a witness for its reflexivity,
which means that we have $\refl_a : (a = a)$ for each $a : A$
(we will sometimes omit the subscript).
Its induction principle, called the \emph{J-rule}, states that
to produce an element of the type $\Pi(b : A). \Pi (p : a = b). C(b, p)$ it suffices
to provide an element of $C(a, \refl_a)$. \\
Caveat: In \href{http://www.cs.nott.ac.uk/~psznk/agda/confluence/}{Agda}, the roles of $=$ and $\equiv$ are reversed. Definitions are written using $=$, while the equality type is written using $\equiv$.
\item For $A : \UU$ and $B : A \to \UU$, we will need the type of dependent pairs,
written $\Sigma(a : A). B(a)$, with the non-dependent version (for $C : \UU$)
written $A \times C$.
\item We require the disjoint sum $A \uplus C$ of types,
with obvious functions $\inl : A \to A \uplus C$ and $\inr : C \to A \uplus C$.
\item We assume that the type theory has inductive types and inductive families, examples for which we will see in \cref{sec:closures}.
\item Finally, for $\Pi$-types, we assume function extensionality which, in its simplest form,
says that for two functions $f, g : \Pi(a : A). B(a)$ we have $f = g$
as soon as they are pointwise equal: $\Pi (a : A). f(a) = g(a)$.%
\footnote{It is known that this simple form already implies the seemingly stronger formulations of function extensionality.}
\end{itemize}

To improve readability and usability, many proof assistants 
implement implicit arguments.
We use them in writing as well, as a purely notational device, as
we write implicit arguments in curly brackets $\{\}$.
For example, $\Pi\{a:A\}. B(a) \to C$ denotes the same type as $\Pi(x:A). B(a) \to C$.
For $a : A$ and $b : B(a)$, if we have $f : \Pi\{a:A\}. B(a) \to C$ 
and $g : \Pi(a:A). B(a) \to C$, we implicitly uncurry and write $g(a, b) : C$ but can omit the implicit
argument in the other case and write $f(b) : C$.
Instead of $\Pi(a:A).\Pi(b: B(a)).C(a,b)$, we write $\Pi(a:A),(b:B(a)). C(a,b)$.

The inductive definition of equality induces the structure of a (higher) groupoid
on types, which forms the basis of the synthetic kind of topology performed in
homotopy type theory.
We will not introduce the corresponding terminology in full, but we will in the
following go through the notions we will use in this paper.

Besides transitivity and symmetry, equality has the property that any function
becomes a functor with respect to the induced groupoids:
Given $f : A \to B$, $a, a' : A$ and an equality $p : a = a'$, we have
$\ap_f(p) : f(a) = f(a')$.
Also, equal elements are indiscernible in the sense that if $B : A \to \UU$ is
a type family over $A$ and we have $a, a' : A$ with $p : a = a'$,
we have the implication $p_* : B(a) \to B(a')$, the so called
\emph{transport} along $p$.

Taking iterated equalities (i.e.\ equalities on equality types), we discover a
type's \emph{higher structure}.
The number of iterations of equality types we have to take after which equalities
don't carry information is called a types \emph{h-level}.
In particular, a type $A : \UU$ in which we have
\begin{itemize}
\item $\Pi (x, y : A). x = y$ is called a \emph{proposition} or \emph{(-1)-type},
\item $\Pi (x, y : A) (p, q : x = y). p = q$ is called an \emph{(h-)set} or \emph{0-type},
\item $\Pi (x, y : A) (p, q : x = y) (\alpha, \beta : p = q). \alpha = \beta$ is called a \emph{1-type}, and so on.
\end{itemize}

Finally, another notion which makes an appearance in this section paper is the
one of \emph{equivalent types}.
We denote with $A \simeq B$ the type of \emph{equivalences} between $A$ and $B$,
i.e.\ functions $f : A \to B$ which have a two-sided inverse $g : B \to A$ such
that for each $x : A$ the two ways of proving the equality $f (g (f (x))) = x$
(by cancelling $g \circ f$ and by cancelling $f \circ g$) coincide.

We refer to \cite{hott-book} for the details of the concepts discussed above.

\subsection{Properties and closures of binary relations\nopunct}
\label{sec:closures}
 (files \href{http://www.cs.nott.ac.uk/~psznk/agda/confluence/squier.graphclosures.html}{\nolinkurl{graphclosures}}, \href{http://www.cs.nott.ac.uk/~psznk/agda/confluence/squier.accessibility.html}{accessibility}, and \href{http://www.cs.nott.ac.uk/~psznk/agda/confluence/squier.listextension.html}{listextension})\textbf{.}
By a \emph{binary relation} on a type $A$ we mean, in the type-theoretic setting, simply a type family $R : A \to A \to \UU$.
This is sometimes called a \emph{proof-relevant} binary relation since $R \, a \, b$ can potentially have many (different) inhabitants.

Thus, a \emph{1-polygraph} in type theory is simply a type $\wasSigmaNull : \UU$ together with a binary relation $(\_ \torel \_) : \wasSigmaNull \to \wasSigmaNull \to \UU$, where the blanks reserve the places for arguments; that is, we write $(x \torel y) : \UU$, mirroring the notation used in \cref{sec:non-tt}.
As before, we will write
\begin{itemize}
 \item $\leadstostar$ for the reflexive-transitive closure,
 \item $\leadstofrom$ for the symmetric closure,
 \item $\leadstofromstar$ for the symmetric-reflexive-transitive closure
\end{itemize}
of the relation $\torel$.
The type-theoretic implementation of these closures is standard.
The reflexive-transitive closure is constructed as an inductive family with two constructors:
\begin{equation} \label{eq:refl-trans-rel-def}
 \begin{aligned}
  & \text{inductive } (\_ \leadstostar \_) : \, \wasSigmaNull \to \wasSigmaNull \to \UU \text{ where} \\
  & \qquad \mathsf{nil}: \Pi (x : \wasSigmaNull).\, (x \leadstostar x) \\
  & \qquad \consLong: \Pi\{x, y, z : \wasSigmaNull\}.\, (x \leadstostar y) \to (y \leadsto z) \to (x \leadstostar z)
 \end{aligned}
\end{equation}
The symmetric closure is the obvious disjoint sum,
\begin{equation} \label{eq:sym-rel-def}
 (x \leadstofrom y) \defeq (x \leadsto y) \uplus (y \leadsto x).
\end{equation}
The symmetric-reflexive-transitive closure $\leadstofromstar$ is constructed by
first taking the symmetric and then the reflexive-transitive closure.
The transitive closure $\leadsto^+$ is defined analogously to $\leadstostar$, but with $\mathsf{nil}$ replaced by a constructor taking a single step (cf.\ 
\href{http://www.cs.nott.ac.uk/~psznk/agda/confluence/squier.graphclosures.html#1590}{TransReflClosure},
\href{http://www.cs.nott.ac.uk/~psznk/agda/confluence/squier.graphclosures.html#1258}{SymClosure},
\href{http://www.cs.nott.ac.uk/~psznk/agda/confluence/squier.graphclosures.html#15280}{TransClosure} in the formalisation).

Since the relation $\leadsto$ is proof-relevant, the same is the case for its closures.
As a consequence, we have functions $(x \leadstostar y) \to (x \leadsto^{**} y)$
and $(x \leadsto^{**} y) \to (x \leadstostar y)$, but these are in general not inverse to each other.
The analogous caveat holds for the other closure operations.
However, this observation has no consequences for our constructions and proofs.

Mirroring the earlier terminology, we call an element of $x \tostar y$ a \emph{sequence}, and an inhabitant of $x \tofromstar y$ a \emph{zig-zag}.
Let us write $\varepsilon_x$ instead of $\mathsf{nil} \, x$ for the trivial sequence at point $x$.
Given two sequences (or zig-zags) $u : x \tostar y$ and $v : y \tostar z$, the definition of their concatenation (by induction on $u$) is standard.
Adopting the notation in \cref{sec:non-tt}, we write $u \cdot v$ for this concatenation, no matter whether $u, v$ are single steps, sequences, or elements of the symmetric closure.
It is standard that the operation $\mathop{\cdot}$ is associative.
Moreover, $u : x \leadstofromstar y$ can be inverted by inverting every single step, and we denote this operation by $u^{-1} : y \leadstofromstar x$.
Inversion and concatenation interact in the obvious way.
Note that $x \leadstostar y$ embeds into $x \leadstofromstar y$ as a \emph{positive}
zig-zag while $y \leadstostar x$ embeds into $x \leadstofromstar y$ as a \emph{negative}
zig-zag, this giving us two distinct embeddings in the case of a closed zig-zag
$x \leadstofromstar x$.

Since it constitutes a real difference between the type-theoretic development 
and the set-theoretic on in \cref{sec:non-tt}, we make the following definition
and statements explicit:

\begin{definition}[length of sequences or zig-zags; cf.\ \href{http://www.cs.nott.ac.uk/~psznk/agda/confluence/squier.graphclosures.html\#6664}{length$^\text{t}$}] \label{def:length-of-zig-zag}
	There is an obvious function
	\begin{equation}
	\length: \Pi\{x, y : \wasSigmaNull\}.\, (x \tostar y) \to \N
	\end{equation}
	which calculates the length of a sequence.
	Starting with $\tofrom$ instead of $\torel$, it calculates the length of a zig-zag.
	We say that a sequence (zig-zag) $\alpha$ is \emph{empty} if its length is zero, and write
	\begin{align}
	& (x \tozero x) \; \defeq \; \Sigma(u : x \tostar x). \length (u) = 0 \\
	& (x \tofromzero x) \; \defeq \; \Sigma(u : x \tofromstar x). \length (u) = 0.
	\end{align}
\end{definition}

The trivial closed sequence (or zig-zag) $\varepsilon_x : (x \tostar x)$ is, of course, empty,
but in the case where $A$ is not a set but a higher type,
not every empty closed sequence (zig-zag) is equal to $\varepsilon_x$.
Instead, an empty closed sequence (zig-zag) at point $x$ corresponds to a path of type $(x = x)$ in $\wasSigmaNull$.
Since \eqref{eq:refl-trans-rel-def} without the second
constructor is the usual definition of Martin-L\"of's identity type as an inductive family, it is easy to see the following:

\begin{lemma} \label{lem:three-types-equivalent}
	For any point $x : \wasSigmaNull$, the three types $(x \tozero x)$ and $(x \tofromzero x)$ and $(x = x)$ are equivalent. \qed
\end{lemma}
\begin{corollary}\label{cor:empty-versus-trivial}
	If $\wasSigmaNull$ is a set, then the trivial sequence (or zig-zag) $\varepsilon_x$ is the only empty sequence (or zig-zag) at point $x$. \qed
\end{corollary}

Given a zig-zag $u : x \tofromstar y$, we can ask whether it is \emph{strictly increasing (decreasing)}, i.e.\ whether each step in the zig-zag comes from the left (right) summand in \eqref{eq:sym-rel-def},
i.e.\ is of the form $\mathsf{inl}(t)$ with $t : x \leadsto y$ ($\mathsf{inr}(t)$ with $t : y \leadsto x$).
The terminology is then the obvious one, copying the one from the set-theoretic development:
We call a zig-zag a \emph{peak} if it is the concatenation of an increasing and a decreasing
zig-zag, a \emph{valley} if it is the concatenation of a decreasing and an
increasing zig-zag, and so on.

The next concept which has an interesting equivalent in type theory is the
one of wellfoundedness or Noetherianness.
As remarked before, the usual (classical) formulation of the form ``no infinite sequence exists'' is unsuitable in a constructive setting.
The inductive characterisation which we use instead is well-known in type theory and due to Aczel~\cite{aczel2001notes}.
Given a binary relation $<$ on a type $A$, one first defines
the notion of accessibility as an inductive type family:
\begin{definition}[$\Phi_<$ in \cite{aczelinductive};
\href{http://www.cs.nott.ac.uk/~psznk/agda/confluence/Cubical.Induction.WellFounded.html\#318}{Acc} in the cubical library]
 The family $\acc < : A \to \UU$ is generated inductively by a single constructor,
 \begin{equation}
     \mathsf{step}: \Pi(a:A).\, (\Pi(x:A).\, (x < a) \to \acc < (x)) \to \acc < (a),
 \end{equation}
 i.e.\ an element $a$ is \emph{accessible} ($\acc \, a)$ if every element smaller than it is.
 The relation $<$ is \emph{wellfounded} if every element is accessible,
 \begin{equation}
  \mathsf{isWellFounded}(<) \defeq \Pi(a:A).\, \acc < (a).
 \end{equation}
 If $<$ is wellfounded, then $>$, where $(x > y) \defeq (y < x)$, is called \emph{Noetherian}.
\end{definition}

While the definition in \cite[Chp.~10.3]{hott-book} is only given for the special
case that $A$ is a set and $<$ is valued in propositions, the more general case
that we consider works in exactly the same way (cf.\ our formalisation).
In particular, we have the following two results:

\begin{lemma}[cf.\ \href{http://www.cs.nott.ac.uk/~psznk/agda/confluence/Cubical.Induction.WellFounded.html\#501}{isPropAcc} in the cubical library]
 For any $x$, the type $\acc < (x)$ is a proposition.
 Further, the statement that $<$ is wellfounded is a proposition. \qed
\end{lemma}

Proving properties of wellfounded relations is made possible by the following principle:
\begin{lemma}[{accessibility induction \cite[Chp.~10.3]{hott-book}; \href{http://www.cs.nott.ac.uk/~psznk/agda/confluence/squier.accessibility.html\#770}{acc-ind}}]\label{lem:tt-acc-induction}
 Assume we are given a family $P : A \to \UU$ such that we have
 \begin{equation}
  \Pi(a_0 : A).\, \acc < (a_0) \to \left(\Pi(a<a_0).\, P(a)\right) \to P(a_0).
 \end{equation}
 In this case, we get:
 \begin{equation}
  \Pi(a_0 : A).\, \acc < (a_0) \to P(a_0).
 \end{equation}
 If $(<)$ is wellfounded, the argument $\acc < (a_0)$ can be omitted and the
 principle is known as \emph{wellfounded induction}.
\end{lemma}

An easy application which demonstrates this induction principle is the following:
\begin{lemma} \label{lem:acc-trans-acc}
Let us write $<^+$ for the transitive closure of $<$. If $a:A$ is $<$-accessible, then it is $<^+$-accessible.
\end{lemma}
\begin{proof}
	By $<$-accessibility induction on $P(a) \defeq (\acc < (a) \to \acc {<^+} (a)$.
\end{proof}
\begin{corollary}[of \cref{lem:acc-trans-acc}; cf.\
\href{http://www.cs.nott.ac.uk/~psznk/agda/confluence/squier.accessibility.html\#3052}{transitive-wellfounded}] \label{cor:+-Noetherian}
	If $<$ is well-founded, then so is $<^+$. \qed
\end{corollary}

\emph{Nested} induction takes the following form:
\begin{lemma}[nested accessibility/wellfounded induction, cf.\ \href{http://www.cs.nott.ac.uk/~psznk/agda/confluence/squier.accessibility.html\#1321}{double-acc-ind}] \label{lem:nested-ind}
 Assume we are given a relation $<_1$ on a type $B$, a relation $<_2$ on a type $C$,
 and a family $P : B \times C \to \UU$.
 Assume further that we are given
 \begin{equation}\label{eq:nested-ind-assumption}
 \begin{aligned}
  & \Pi(b : B), (c : C).\, \acc {<_1} (b) \to \acc {<_2} (c) \to \\
  & \qquad (\Pi(b' <_1 b).\, P(b',c)) \to (\Pi(c'<_2 c).\, P(b,c')) \to \\
  & \qquad P(b,c).
 \end{aligned}
 \end{equation}
 Then, we get:
 \begin{equation}
  \Pi(b : B), (c : C).\, \acc {<_1} (b) \to \acc {<_2} (c) \to P(b,c).
 \end{equation}
\end{lemma}
\begin{proof}
We carefully apply accessibility induction with the correct motive on the
witnesses of both $\acc {<_1} (b)$ and $\acc {<_2} (c)$,
then apply \eqref{eq:nested-ind-assumption}.
\end{proof}

\cref{lem:nested-ind} is needed for the type-theoretic proof of
the property \ref{enum:comp-acc} in the proof of \cref{lem:list-ext-noetherian}.
The other parts are identical.
This means, again, that the list extension of a wellfounded relation is wellfounded. The Agda formalisation of this fact can be found in the file \href{http://www.cs.nott.ac.uk/~psznk/agda/confluence/squier.listextension.html}{listextension}, with the main theorem of that module being \href{http://www.cs.nott.ac.uk/~psznk/agda/confluence/squier.listextension.html#4828}{isWF\textlangle$>$\textrangle$\Rightarrow$isWF\textlangle$>^\text{L}$\textrangle}.

Of course, all statements about wellfounded relations dualise in the obvious way to Noetherian relations; in particular,
the list extension of a Noetherian relation is Noetherian.
However, proving this in Agda requires a certain amount of work.
For example, we show that a monotone function between types with orders reflects accessibility
(cf.\ \href{http://www.cs.nott.ac.uk/~psznk/agda/confluence/squier.accessibility.html#3434}{acc-reflected}), and that reversing a list commutes (in a weak sense) with taking its transitive closure
(cf.\ \href{http://www.cs.nott.ac.uk/~psznk/agda/confluence/squier.graphclosures.html#16051}{revClosureComm}).
Only then, we are able to draw the seemingly obvious conclusion from 
\cref{cor:+-Noetherian}
that the transitive closure of a Noetherian relation is Noetherian (cf.\ \href{http://www.cs.nott.ac.uk/~psznk/agda/confluence/squier.accessibility.html\#4274}{transitive-Noetherian}).

We have already seen an example of a Noetherian relation in the introduction of this paper, namely 
the relation \eqref{eq:fg-relation} that is used in the construction of the free group in \cref{ex:fg}:
Given two lists $\ell_1,\ell_2 : \List(M \uplus M)$, we have $\ell_1 \torel \ell_2$ if the first list can be transformed into the second list by removing exactly two elements.
The two removed list elements have to be consecutive and ``inverse'' to each other, i.e.\ one is of the form $\inl(a)$, the other $\inr(a)$.
It is clear that this relation is Noetherian, since each step reduces the length of the list.

\begin{lemma}[{free groups, continuing \cref{ex:fg}}] \label{lem:fg-Noether-NO-confl}
	The relation $\torel$ on lists defined by \eqref{eq:fg-relation}
	is Noetherian. \qed
\end{lemma}

\subsection{2-polygraphs in homotopy type theory\nopunct}
\label{sec:hott-polygraphs}
 (files
\href{http://www.cs.nott.ac.uk/~psznk/agda/confluence/squier.polygraphs.html}{polygraphs},
\href{http://www.cs.nott.ac.uk/~psznk/agda/confluence/squier.cancelInverses.html}{cancelInverses},
\href{http://www.cs.nott.ac.uk/~psznk/agda/confluence/squier.newman.html}{newman},
\href{http://www.cs.nott.ac.uk/~psznk/agda/confluence/squier.homotopybasis.html}{homotopybasis})\textbf{.}
With these necessary notions defined, we can now translate generalised 2-polygraphs.
Here, it is important that we will not view reduction steps and rewrite steps as
plain types but always as parameterised by their source and target, see \cref{rem:fibration-vs-families}.
\begin{definition}[2-polygraph, cf.\ \href{http://www.cs.nott.ac.uk/~psznk/agda/confluence/squier.polygraphs.html\#849}{2-polygraph}]\label{def:tt-2-polygraph}
A generalised 2-polygraph is a triple $\Sigma \equiv (\wasSigmaNull, \torel, \To)$
consisting of data of the following types:
\begin{itemize}
\item A type $\wasSigmaNull : \UU$ of objects of the 2-polygraph,
\item a family $(\torel) : \wasSigmaNull \to \wasSigmaNull \to \UU$ of reduction steps, and
\item a family 
$(\To) : \Pi\{x, y : \wasSigmaNull\}.\, (x \tofromstar y) \to (x \tofromstar y) \to \UU$
of rewrite steps of $\Sigma$.
\end{itemize}
We might  denote with $\Sigma_1$ and $\Sigma_2$ the total spaces of the type
families of reduction steps and rewrite steps, respectively, in order to recover the content of \cref{subsec:gen-2-poly}.
Transferring  the properties of 2-polygraphs introduced before into the realm
of type theory is straightforward.
Note that none of these definitions are necessarily propositional, so they carry
\emph{data} instead of mere proofs.
\begin{itemize}
\item The 1-polygraph $(\Sigma_0, \Sigma_1)$
is called \emph{terminating} (cf.\ \href{http://www.cs.nott.ac.uk/~psznk/agda/confluence/squier.polygraphs.html#4139}{isNoetherian}) if it comes with
transitive Noetherian relation $(>)$ on $\wasSigmaNull$ and a function
\begin{equation} 
(x \torel y) \to (x > y) \text{.}
\end{equation}
\item $\Sigma$ is called \emph{closed under congruence} (cf.\ \href{http://www.cs.nott.ac.uk/~psznk/agda/confluence/squier.polygraphs.html#4871}{$\Leftrightarrow^*\!\!\text{isCongrClosed}$}) if for
$u : (x \tofromstar y)$, $v : (y' \tofromstar z)$, and $\alpha : (w \Tofromstar w')$
for some $w, w' : (y \tofromstar y')$ we have
\begin{equation}
u \cdot \alpha \cdot v : (u \cdot w \cdot v \Tofromstar u \cdot w' \cdot v) \text{.}
\end{equation}
\item $\Sigma$ has a \emph{Winkler-Buchberger structure} (cf.\ 
\href{http://www.cs.nott.ac.uk/~psznk/agda/confluence/squier.newman.html#1316}{hasWB})
if, for each local
peak $u : (y \from x \torel z)$, we have a $u' : (y \tofromstar z)$ such that $x > x_i$ for any object occurring in the inner part of $u'$, together with an element
$\WB(u) : (u \Tofromstar u')$.
\item $\Sigma$ \emph{cancels inverses} (cf.\
\href{http://www.cs.nott.ac.uk/~psznk/agda/confluence/squier.cancelInverses.html#830}{cancels*RInv}) if we supply a function
\begin{align}
\LINV &: \Pi\{x, y : \wasSigmaNull\},(u : x \tofromstar y).\, (u \cdot u^{-1} \Tostar \varepsilon_x).
\end{align}
\end{itemize}
\end{definition}

\begin{remark} \label{rem:>-could-be-truncation}
An important special case, sufficient for the applications in \cref{sec:applications}, is the case where $\torel$ is Noetherian and $>$ can simply be chosen to be the transitive closure of $\torel$ (cf.\ \cref{thm:unary-basis}).
Although we do not have an example where it happens, it is plausible that there may be cases where one should take $>$ to be the propositional truncation of $\torel^+$ instead: Everything will work in exactly the same way, but the requirement on the Winkler-Buchberger structure will be weaker.
\end{remark}

The other notions of confluence structures can be defined in a straightforward way.
By the same construction as in \Cref{sec:non-tt}, a Winkler-Buchberger structure
induces a Church-Rosser structure on the 2-polygraph:
\begin{lemma}[translation of \Cref{lem:wb-to-cr}, cf.\ module \href{http://www.cs.nott.ac.uk/~psznk/agda/confluence/squier.newman.html\#3757}{wb2cr}]
From a terminating generalised 2-polygraph which is closed under congruence and has a Winkler-Buchberger structure $\WB$,
we can construct for each $u : (y \tofromstar z)$ a valley
$u' : (y \leadstostar x' \leadsfromstar z)$ and a rewrite zig-zag $\CR(u) : (u \Tofromstar u')$.
Again, we will call $x'$ the \emph{reduct} of $\CR(u)$. \qed
\end{lemma}

A homotopy basis in type theory becomes just a single inhabitant of a $\Pi$-type:
\begin{definition}[cf.\ \href{http://www.cs.nott.ac.uk/~psznk/agda/confluence/squier.homotopybasis.html\#1019}{hasHomotopyBasis}]
A \emph{homotopy basis} of $\Sigma$ is a function
\begin{equation}
\alpha : \Pi\{x \, y : \wasSigmaNull\},(u \, v : (x \tofromstar y)). (u \Tofromstar v) \text{.} 
\end{equation}
\end{definition}

When collecting all the structures we need on our 2-polygraph to construct a homotopy
basis, the list of assumptions in \cref{thm:main-theorem-settheory}
is insufficient for the reason explained in \cref{subsec:translation-is-generalisation} and made precise in \cref{lem:three-types-equivalent}.
Recall that, in the proof of \cref{thm:main-theorem-settheory},
we had to consider the case that $u$ is a closed zig-zag at point $y$ with $s_1(\CR(u))$ an empty zig-zag.
In contrast to before, this case is not trivial in the type-theoretic setting with higher equalities.
What we need it the following property:

\begin{definition}[cf.\ \href{http://www.cs.nott.ac.uk/~psznk/agda/confluence/squier.homotopybasis.html\#1070}{cancelsEmpty}]
A 2-polygraph is said to \emph{cancel empty closed zig-zags} if we have a function
\begin{equation*}
e : \Pi\{x : \wasSigmaNull\},(u : (x \tofromstar x)). (\length(u) = 0) \to (u \Tofromstar \varepsilon_x) \text{.}
\end{equation*}
\end{definition}
By \cref{cor:empty-versus-trivial}, any 2-polygraph $\Sigma$, where the type of objects $\wasSigmaNull$ is a set, cancels empty zig-zags trivially.

With the above definitions at hand, the translation of the construction of the
homotopy basis is straightforward:

\begin{theorem}[translation of \cref{thm:main-theorem-settheory}, cf.\ \href{http://www.cs.nott.ac.uk/~psznk/agda/confluence/squier.homotopybasis.html\#6351}{Noeth×WB×Congr×Cancel$\To$Basis}]\label{thm:tt-basis}
Let $\Sigma$ be a 2-polygraph which is terminating, closed under
congruence, cancels inverses, has a Winkler-Buchberger structure, and
cancels empty closed zig-zags.
Then, $\Sigma$ has a homotopy basis.
\end{theorem}

\begin{proof}
By and large, the proof proceeds the same way as the one presented for 
\cref{thm:main-theorem-settheory}:
We apply Noetherian induction (cf.~\cref{lem:tt-acc-induction}) on the
type $\wasSigmaNull$
with
\begin{equation*}
P(x) \defeq \Pi(u : x \tofromstar x). (u \Tofromstar \varepsilon_x) \text{.}
\end{equation*}
As before, we fix $x : \wasSigmaNull$ and $u : (x \tofromstar x)$ and
may assume that we are given $\alpha_v := P(y)(v)$ for all $x > y$ and
$v : (y \tofromstar y)$.
We again consider $\CR(u) : (u \Tostar v \cdot w^{-1})$ with
$v, w : (x \leadstostar y)$ for the reduct $y$ of $\CR(u)$.
If $v$ or $w$ (and thus also the other) have length zero, 
then there is $e_v : (v \Tofromstar \varepsilon_y)$, allowing us to construct the
rewrite chain $u \Tostar v \cdot w^{-1} \Tofromstar \varepsilon_x$.

In the case where either $v$ or $w$ have non-zero length, we again can conclude
that $x > y$ and that the induction hypothesis can be applied as in the set-theoretic formulation.
\end{proof}

\subsection{Noetherian induction for closed zig-zags\nopunct} \label{subsec:noetherian-cycle-induction}
(file \href{http://www.cs.nott.ac.uk/~psznk/agda/confluence/squier.noethercycle.html}{noethercycle})\textbf{.}
While so far we considered the rewrite zig-zags of a 2-polygraph mainly as
additional data to a 1-polygraph,
we can also think of them as a witness of the fact that certain properties
of reduction zig-zags carry over to other, rewritten, reduction zig-zags.
This separation between structured data and statements is more blurred in a type
theoretic setting than in a set theoretic one, but we can still use the
homotopy basis we constructed to say something \emph{about} statements which we
can prove about closed reduction zig-zags.
In type theory, it is natural to phrase such a tool as a lemma reminiscent of an \emph{induction principle},
i.e.\ a scheme to prove statements about composite structures by recursively
proving things about simpler structures.%
\footnote{Depending on the terminology one uses, one may wish to reserve the term \emph{induction principle} for the native elimination principles that inductive types are equipped with. Our result (\cref{thm:unary-basis}) is of course not such an elimination principle, but can be stated in a similar form. Being slightly less strict about the usage of the term, it thus seems reasonable to call it an induction principle.}

So far, we have taken the \emph{globular} point of view which considers
surfaces between paths with the same source and target.
However, our assumptions guarantee that we do not lose generality if we assume that one of the zig-zags is empty, since a rewrite $u \Tofromstar v$ can equivalently be represented as a rewrite $u \Tofromstar \varepsilon$.
In type theory, it is often more convenient to consider this version and formulate a principle with only one instead of two arguments.

Moreover, our \cref{thm:tt-basis} is more general than what is needed for the applications in type theory that we will present in \cref{sec:applications}. For the sake of simplicity, we specialise it slightly and formulate the following principle:

\begin{theorem}[Noetherian induction for closed zig-zags, cf.\ the theorem \href{http://www.cs.nott.ac.uk/~psznk/agda/confluence/squier.noethercycle.html\#15255}{induction-for-closed-zigzags}]\label{thm:unary-basis}
Let $(\wasSigmaNull, \torel)$ be a 1-polygraph such that $\torel$ is Noetherian.
Assume further that $\torel$ is locally confluent, i.e.\ that for each local peak $u : (y \from x \torel z)$, there is a valley $\overline u : (y \tostar w \fromstar z)$.

Let $P$ be a type family of the form
\begin{equation*}
P : \Pi\{x : \wasSigmaNull\}.\, (x \tofromstar x) \to \UU \text{,}
\end{equation*}
with the following properties:
\begin{enumerate}[label=\textbf{(\arabic*)},ref={(\arabic*)}]
\item\label{enum:fill-empty}
$P(e)$ holds for all $e : (x \tofromstar x)$ with $\length(e) = 0$.
\item\label{enum:fill-inverses}
For all rewrite zig-zags $u : (x \tofromstar y)$,
we have $P(u \cdot u^{-1})$.
\item\label{enum:fill-rotate}
$P$ is closed under \emph{rotation}:
For $u : (x \tofromstar y)$ and $v : (y \tofromstar x)$, we have
\begin{equation*}
P(u \cdot v) \to P(v \cdot u) \text{.}
\end{equation*}
\item\label{enum:fill-concat}
$P$ is closed under \emph{pasting} of closed zig-zags:
If $u,v,w : (x \tofromstar y)$ are parallel zig-zags, we have
\begin{equation*}
P(u \cdot v^{-1}) \to P(v \cdot w^{-1}) \to P(u \cdot w^{-1}) \text{.}
\end{equation*}
\item\label{enum:fill-inv}
$P$ is closed under \emph{inversion} of closed zig-zags:
For $u : (x \tofromstar x)$ we have 
\begin{equation*}
P(u) \to P(u^{-1})\text{.}
\end{equation*}
\item\label{enum:fill-wb}
$P$ holds for the ``outlines'' of the chosen local confluence diamonds,
i.e.\ for every local peak $u$,
we have
$P(u \cdot \overline u^{-1})$.
\end{enumerate}
Then, $P$ holds for all closed reduction zig-zags:
\begin{equation*}
\Pi\{x : \wasSigmaNull\},(u : x \tofromstar x).\, P(u) \text{.}
\end{equation*}
\end{theorem}

\begin{proof}
The given data give rise to the 2-polygraph
$(\wasSigmaNull, \leadsto, \To)$, where the family of reduction steps 
between $u, v : (x \tofromstar y)$ is defined by
\begin{equation}
(u \To v) \defeq P\left(u \cdot v^{-1}\right) \text{.}
\end{equation}

First, we observe that the definition, together with the conditions on $P$, imply
that we have
\begin{equation}
(u \Tofromstar v) \to P\left(u \cdot v^{-1}\right) \text{.}
\end{equation}
To see this, we apply induction on the length of reduction a zig-zag
$\alpha : (u \Tofromstar v)$:
\begin{itemize}
\item If $\alpha$ has length $0$, we are done by \ref{enum:fill-inverses}.
\item If $\alpha$ starts with a reduction $\beta : (u \To w)$, we obtain
$P(w \cdot v^{-1})$ by induction from the remaining reduction zig-zag
and $P(u \cdot w^{-1})$ from $\beta$ by definition, from which
\ref{enum:fill-concat} lets us deduce $P(u \cdot v^{-1})$.
\item Lastly, if $\alpha$ instead starts with a reversed reduction step $\beta : (w \To u)$,
we use \ref{enum:fill-inv} and are again in the situation of the previous case.
\end{itemize}
Therefore, a homotopy basis for our 2-polygraph
$(\wasSigmaNull, \leadsto, \To)$ suffices to complete the proof.
We check the conditions of \cref{thm:tt-basis}:
\begin{itemize}
\item \textbf{Cancellation of inverses:} Note that the requirement $s \cdot s^{-1} \Tostar \varepsilon$ follows from $s \cdot s^{-1} \To \varepsilon$, which is the same as $s \To s$, which (by the above observation) is implied by $s \Tostar s$, which is trivial. Alternatively, cancellation of inverses is also a direct consequence of \ref{enum:fill-inverses}.
\item \textbf{Closure under congruence:}
To show that $(\To)$ is closed under congruence, let $u : (x \tofromstar y)$ and
$v : (y' \tofromstar z)$ be given together with $w, w' : (y \tofromstar y')$, for which
we further assume $\alpha : (w \To w')$.
We need to show that
\begin{equation*}
P\left((u \cdot w \cdot v) \cdot (u \cdot w' \cdot v)^{-1}\right) \text{,}
\end{equation*}
which by rotation follows from
$P\left(u^{-1} \cdot u \cdot w \cdot v \cdot v^{-1} \cdot w'^{-1}\right)$.
Note that the special case of \ref{enum:fill-concat} with $v \equiv \varepsilon$
allows us to concatenate closed zig-zags.
In particular, is suffices to prove
$P\left(u^{-1} \cdot u\right)$
and $P\left(w \cdot v \cdot v^{-1} \cdot w'^{-1}\right)$.
The former holds by the discussed cancellation of inverses.
For the latter, we use the same trick again to write it as concatenation of
$P\left(v \cdot v^{-1}\right)$ and $P\left(w'^{-1} \cdot w\right)$.
This time, the second part holds by rotation of the assumption $\alpha$.
\item \textbf{Winkler-Buchberger structure:} The local confluence structure is also a Winkler-Buchberger structure.%
\footnote{While it seems fair to consider this connection between a local confluence structure and a Winkler-Buchberger structure obvious, we found it rather cumbersome to formalise is. The corresponding Agda statement is more than 70 lines of code long, and this number does not even take into account that we had to write a range of auxiliary lemmas solely for this particular statement.}
\item \textbf{Cancellation of empty zig-zags:} This is almost directly given by \ref{enum:fill-empty}.
\end{itemize}
As pointed out above, the thereby obtained homotopy basis completes the construction of $P(u)$ for every $u : (x \tofromstar x)$.
\end{proof}

\begin{remark} \label{rem:lics-paper-vs-current-article}
We could of course derive an induction principle that is more general than \cref{thm:unary-basis},
where we have given up some generality to obtain simplicity.
On the other hand, \cref{thm:unary-basis} holds even if we remove the assumptions \ref{enum:fill-inverses} and \ref{enum:fill-inv}, i.e.\  $P(u \cdot u^{-1})$ and $P(u) \to P(u^{-1})$, which we have proved in our conference paper \cite{krausVonRaumer:wellfounded}.
Unfortunately, if we want to present \cref{thm:unary-basis} as a special case of \cref{thm:tt-basis},
the assumptions \ref{enum:fill-inverses} and \ref{enum:fill-inv} seem to be unavoidable.
This slight loss in generality stems from the symmetric definition of a homotopy basis in \cref{thm:main-theorem-settheory}, but poses no meaningful restriction in our applications.
\end{remark}

\section{Applications in homotopy type theory}\label{sec:applications}

From now on, we are working in ``full'' homotopy type theory; i.e.,
in addition to the components detailed in \cref{subsec:specify-type-theory},
we assume that the theory has higher inductive types and families, and that the
universe $\UU$ is univalent:
By assuming that for $A, B : \UU$ the canonical function
\begin{equation*}
(A = B) \to (A \simeq B)
\end{equation*}
is an equivalence itself, we allow ourselves to treat equivalences and equalities
between types as the same.

Besides quotients and pushouts, which will be introduced as we go along,
we need the truncation $\trunc{n}{A}$, which for a type $A : \UU$ and
$n \ge -1$ represents a copy of $A$ which is made an $n$-type by equating all
(higher) equalities above the stated level.
Truncation is defined as a higher inductive type with constructors
\begin{align*}
\iota &: A \to \trunc{n}{A} \text{ and} \\
\mathsf{trunc} &: \istype{n}{\trunc{n}{A}} \text{.}
\end{align*}
Note that the induction principle for $\trunc{n}{A}$ states that the type can
only eliminate into (families of) $n$-types.

After establishing some preliminary results for our applications in \cref{sec:application-preliminaries},
we will, in \cref{sec:quotients}, use the induction principle for closed
zig-zags to characterise functions that map from a quotient into a 1-type (groupoid).
This characterisation will then be applied to show that the fundamental group
of the free group is trivial (\cref{sec:free-groups}) and that the
pushout of 1-types over a set has no non-trivial second homotopy groups (\Cref{sec:pushouts}).
We will then, in \cref{sec:special} compare the two results with each other and put them in a series
of six different applications which are all approximations of open problems
in the field of synthetic homotopy theory.

\subsection{On quotients, coequalisers, and truncation}\label{sec:application-preliminaries}

As explained in the introduction, the \emph{set-quotient} ${\setquot}$ is the
higher inductive type with constructors $\iota$, $\glue$, and $\mathsf{trunc}$,
see \eqref{eq:setquotient}.
The construction can be split into two steps.
Recall from \cref{subsec:intro-hott-application}
that we write $\quot$ for the \emph{untruncated quotient} or \emph{coequaliser}
which has only the constructors $\iota$ and $\glue$, and $\trunc 0 -$
for the \emph{set-truncation} as described above.
\begin{lemma} \label{lem:simu-is-consec}
 For a relation $(\leadsto)$ on $\wasSigmaNull$, we have
 \begin{equation}
  \left({\setquot}\right) \simeq \trunc 0 {\quot}.
 \end{equation}
\end{lemma}
\begin{proof}
 The direct approach of constructing functions back and forth works without difficulties.
\end{proof}
For a given type $X$, there is a canonical map from the function type
$(\quot) \to X$ to the $\Sigma$-type of pairs $(f,h)$, where
\begin{align}
 & f : \wasSigmaNull \to X \label{eq:fAX}\text{,}\\
 & h : \Pi\{x, y : \wasSigmaNull\}. (x \leadsto y) \to f(x) = f(y). \label{eq:hAX}
\end{align}
This map is given by:
\begin{equation} \label{eq:canonicalmap}
 g \mapsto (g \circ [-], \ap_g \circ \glue).
\end{equation}
The universal property of the higher inductive type $\quot$ tells us that this
function is an equivalence (one can of course also show this with the dependent
elimination principle of $\quot$, if that is assumed instead as primitive).

We will need to prove statements about equalities in coequalisers.
For this, we use the following result which characterises the path spaces of $(\quot)$: 

\begin{theorem}[induction for coequaliser paths,~\cite{KrausVonRaumer_pathSpaces}] 
\label{thm:lics2019-main}
 Let a relation $(\leadsto) : \wasSigmaNull \to \wasSigmaNull \to \UU$ as before and a
 point $x_0: \wasSigmaNull$ be given.
 Assume we further have a type family
 \begin{equation} \label{eq:mainresult-based-P}
 P: \Pi \{y : \wasSigmaNull\}.(\iota(x_0) =_{\quot} \iota(y)) \to \UU
 \end{equation}
 together with terms
  \begin{align}
   r &: P(\refl_{\iota(x_0)})\text{,} \\
   e &: \Pi\{y, z : \wasSigmaNull\}, (q : \iota(x_0) = \iota(y)), (s : y \leadsto z).
         P(q) \simeq P(q \ct \glue(s))\text{.}
  \end{align}
 Then, we can construct a term
 \begin{equation}
  \mathsf{ind}_{r,e} : \Pi\{y : \wasSigmaNull\},(q : \iota(x_0) = \iota(y)). P(q)
 \end{equation}
 with the following $\beta$-rules:
 \begin{align}
  & \mathsf{ind}_{r,e} (\refl_{\iota(x_0)}) = r \label{eq:thm-based-first-beta}\text{,} \\
  & \mathsf{ind}_{r,e}(q \ct \glue(s)) = e (q,s, \mathsf{ind}_{r,e}(q))\text{.} \label{eq:thm-based-second-beta}
 \end{align}
 \qed
\end{theorem}

To have all prerequesites we need in order to characterise the functions out
of a quotient, we need one additional characterisation of maps:
The type of functions from $\trunc 0 {\wasSigmaNull}$ type into a 1-type consists of
those functions which map any loop $p$ to the reflexivity witness in $B$:
For types $\wasSigmaNull$ and $B$, we have a canonical function
\begin{equation} \label{eq:comp-with-tproj}
 (\trunc 0 {\wasSigmaNull} \to B) \to (\wasSigmaNull \to B)
\end{equation}
which is given by precomposition with $\tproj 0 -$.
Any such function $g \circ \tproj 0 -$ is moreover constant on loop spaces in
the sense that
\begin{equation}
 \ap_{g \circ \tproj 0 -} : (x = x) \to (g(x) = g(x))
\end{equation}
satisfies $\ap_{g \circ \tproj 0 -}(p) = \refl$, for all $x$ and $p$.
For a $1$-truncated type $B$, the following known result by Capriotti, Kraus,
and Vezzosi states that this property is all one needs to reverse \eqref{eq:comp-with-tproj}:

\begin{theorem}[\cite{capKraVez_elimTruncs}] \label{lem:CSL}
 Let $A$ be a type and $B$ be a 1-truncated type.
 The canonical function from $(\trunc 0 {A} \to B)$ to the type 
 \begin{equation}
  \Sigma(f : A \to B). \Pi(a : A),(p : a = a). \ap_f(p) = \refl
 \end{equation}
 is an equivalence. \qed
\end{theorem}

\subsection{A characterisation of functions on quotients}\label{sec:quotients}

As before, let $(\leadsto)$ be a relation on $\wasSigmaNull$.
Assume further that we are given a function $f : \wasSigmaNull \to X$ and a proof $h$ 
that $f$ sends related points to equal points, as in \eqref{eq:fAX} and \eqref{eq:hAX}.
There is an obvious function
\begin{equation}
 h^{*} : \Pi\{x,y : \wasSigmaNull\}. (x \tofromstar y) \to f(x) = f(y),
\end{equation}
defined by recursion on $x \tofromstar y$ which in each step composes
with a path given by $h$ or the inverse of such a path.
Given $(f,h)$ and a third map $k : X \to Y$, it is easy to prove by induction on
$x \tofromstar y$ that we have
\begin{equation} \label{eq:star-shifting}
 \ap_k \circ h^{*} = (\ap_k \circ h)^{*}.
\end{equation}
We also note that, for chains $u, v$,
\begin{align}
 & h^{*}(u \cdot v) = h^{*}(u) \ct h^{*}(v)\label{eq:star-inverting} \text{ and}\\
 & h^{*}(u^{-1}) = (h^{*}(u))^{-1}\label{eq:star-functoriality}\text{.}
\end{align}
Of particular interest is the function
$\glue^{*} : \Pi\{x, y : \wasSigmaNull\}. (x \tofromstar y) \to \iota(x) = \iota(y)$.
It is in general not an equivalence: For example, for $t : x \leadsto x$,
the chain $t \cons t^{-1}$  and the empty chain both get mapped to $\refl$.
Thus, $\glue^{*}$ does not preserve inequality (but see \cref{cor:empty-versus-trivial}).
However, we have the following result:
\begin{lemma}
 The function $\glue^{*} : (x \tofromstar y) \to \iota(x) = \iota(y)$ is surjective.
\end{lemma}
\begin{proof}
 Fixing one endpoint $x_0 : \wasSigmaNull$ and setting
 \begin{align}
  & P : \Pi\{y : \wasSigmaNull\}.(\iota(x_0) = \iota(y)) \to \UU \\
  & P(q) \defeq \trunc{-1}{\Sigma(u : x_0 \tofromstar y).\glue^{*}(u) = q}
 \end{align}
we need to show that, for all $q$, we have $P(q)$.
We use \cref{thm:lics2019-main}, where $r$ is given by the trivial chain.
To construct $e$, we need to prove $P(q) \simeq P(q \ct \glue(s))$ for any
$s : y \leadsto z$.
This amounts to constructing functions in both directions between the types
${\Sigma(u : x_0 \tofromstar y).\glue^{*}(u) = q}$ and
${\Sigma(u : x_0 \tofromstar y).\glue^{*}(u) = q \ct \glue(s)}$, where extending
a chain with $s$ or with $s^{-1}$ is sufficient.
\end{proof}
The following is a ``derived induction principle'' for equalities in coequalisers:
\begin{lemma} \label{lem:simple_ind_paths}
 For a family ${P : \Pi\{x : \quot\}. x=x \to \UU}$ such that each $P(q)$ is a proposition,
 the two types
 \begin{equation}
  \Pi\{x : \wasSigmaNull\},(u : x \leadstofromstar x).\,P(\glue^{*}(u)).
 \end{equation}
 and
 \begin{equation}
  \Pi(c : \quot),(q: a = a).\,P(q)
 \end{equation}
 are equivalent.
\end{lemma}
\begin{proof}
 Both types are propositions, and the second clearly implies the first.
 For the other direction, induction on $c : \quot$ lets us assume that $c$ is of the form
 $\iota(x)$ for some $x : \wasSigmaNull$; the case for the constructor $\glue$ is automatic.
 The statement then follows from the surjectivity of $\glue^{*}$.
\end{proof}

\begin{theorem} \label{thm:gensetquotnew}
 Let $\wasSigmaNull : \UU$ be a type, $(\leadsto) : \wasSigmaNull \to \wasSigmaNull \to \UU$ be
 a relation, and $X : \UU$ be a 1-type.
 Then, the type
 of functions $({\setquot} \to X)$ is equivalent to the type of triples $(f,h,c)$ (a nested $\Sigma$-type), where
 \begin{align}
  & f : \wasSigmaNull \to X \\
  & h : \Pi\{x, y : \wasSigmaNull\}. (x \leadsto y) \to f(x) = f(y) \\
  & c : \Pi\{x : \wasSigmaNull\}(u : x \leadstofromstar x). h^{*}(u) = \refl. \label{eq:genset-last-comp}
 \end{align}
\end{theorem}
\begin{proof}
 We have the following chain of equivalences:
\begin{alignat*}{5}
  &&&& \quad & \phantom{\Sigma} \setquot \to X \\[.3cm]
  \textit{by } \cref{lem:simu-is-consec} &\quad &&\simeq & \quad & \phantom{\Sigma} \trunc{0}{\quot} \to X \\[.3cm]
  \textit{by } \cref{lem:CSL} &&&\simeq  &       &\Sigma g : (\quot) \to X. \\
  &&& & \quad & \phantom{\Sigma} c : \Pi\{x:\quot\},(q : x=x). \ap_g(q) = \refl \\[.3cm]
  \textit{by }\cref{lem:simple_ind_paths}&&&\simeq  &  & \Sigma g : (\quot) \to X. \\
  &&& & \quad & \phantom{\Sigma} c : \Pi\{x : \wasSigmaNull\},(u : x \tofromstar x).
     \ap_g(\glue^{*}(\gamma)) = \refl \\[.3cm]
  \textit{by }\eqref{eq:star-shifting}&&&\simeq  &  & \Sigma g : (\quot) \to X. \\
  &&& & \quad & \phantom{\Sigma} c : \Pi\{x : \wasSigmaNull\},(u : x \tofromstar x).
     (\ap_g \circ \glue)^{*}(\gamma) = \refl \\[.3cm]
  \textit{by }\eqref{eq:canonicalmap}&&&\simeq   &  & \Sigma f : \wasSigmaNull \to X. \\
  &&& & \quad & \Sigma h : \Pi\{x, y : \wasSigmaNull\}. (x \leadsto y) \to f(x) = f(y). \\
  &&&&& \phantom{\Sigma} c : \Pi\{x : \wasSigmaNull\},(u : x \tofromstar x). h^{*}(u) = \refl
 \end{alignat*}
\end{proof}

\begin{remark}
It is an easy exercise to show that the component $c$ in the statement of
\cref{thm:gensetquotnew} can be equivalently replaced by
\begin{equation*}
c' : \Pi\{x, y : \wasSigmaNull\}(u, v : x \leadstofrom x). h^{*}(u) = h^{*}(v)
\end{equation*}
to obtain a binary version of the requirement.
\end{remark}

We are now ready to combine the theory developed in this section with the
construction of the homotopy basis to obtain a full characterisation of
maps from a set-quotient into a one-type.

\begin{theorem} \label{thm:dirsetquot}
 Let $\wasSigmaNull : \UU$ be a type, $(\leadsto) : \wasSigmaNull \to \wasSigmaNull \to \UU$
 a Noetherian and locally confluent relation,
 with the local confluence valley of $u$ denoted by $\overline u$ as in \cref{thm:unary-basis}.
 Further, let $X : \UU$ be a 1-type.
 Then, the type of functions 
 $(\setquot) \to X$
 is equivalent to the type of tuples $(f,h,d_1,d_2)$, where
 \begin{align}
  & f : \wasSigmaNull \to X \\
  & h : \Pi\{x, y : \wasSigmaNull\}. (x \leadsto y) \to f(x) = f(y) \\
  & d_1 : \Pi\{x : \wasSigmaNull\}. \Pi(p : x = x). \ap_f(p) = \refl \label{eq:dir-penu-comp} \\
  & d_2 : \Pi\{x, y, z : \wasSigmaNull\}(u : y \leadsfrom x \leadsto z). 
            h^{*} \left(u \cdot \overline{u}^{-1}\right) = \refl. \label{eq:dir-last-comp}
\end{align}
Further, if $\wasSigmaNull$ is a set, the type $\setquot \to X$ is
equivalent to the type of triples $(f,h,d_2)$.
\end{theorem}
\begin{proof}
 The case for $\wasSigmaNull$ being a set follows immediately from the main statement,
 since the type of $d_1$ becomes contractible.

 For the main statement,
 we want to apply \cref{thm:gensetquotnew}.
 We need to show that the type of $c$ in \eqref{eq:genset-last-comp} is equivalent
 to the type of pairs $(d_1,d_2)$ above.
 Note that they are all propositions.
 From $c$, we immediately derive $(d_1, d_2)$ from \cref{cor:empty-versus-trivial}.
 
 Let us assume we are given $(d_1, d_2)$. We need to derive $c$.
 We want to apply the induction principle given by \cref{thm:unary-basis} with
\begin{equation*}
P(u) \defeq (h^{*}(u) = \refl) \text{.}
\end{equation*}
Now, we need to show the six closure properties of $P$ to complete the proof:
\begin{description}
\item[$P$ is true for empty closed zig-zags] By \cref{lem:three-types-equivalent},
each empty closed zig-zag is the image of a loop in $\wasSigmaNull$ under the equivalence
between empty closed zig-zags, which by $d_1$ is mapped to $\refl$.
\item[$P$ is true for concatenation of inverses] For $s : (x \leadsto y)$ the type
\begin{equation*}
h^*(s \cdot s^{-1}) \equiv h(s) \cdot h(s)^{-1} = \refl
\end{equation*}
is inhabited. For longer zig-zags, the statement follows by induction.
\item[$P$ is closed under rotation] For $u$ and $v$ we have
\begin{equation*}
\begin{alignedat}{9}
&P(u \cdot v) &\quad&\equiv&\quad& (h^{*}(u \cdot v) = \refl) \\
&&&\simeq && (h^{*}(u) = h^{*}(v)^{-1}) \\
&&&\simeq && (h^{*}(v) = h^{*}(u)^{-1}) \\
&&&\equiv && P(v \cdot u) \text{.}
\end{alignedat}
\end{equation*}
\item[$P$ is closed under pasting]
We can calculate
\begin{equation*}
\begin{alignedat}{9}
&h^{*}(u \cdot w^{-1}) &\quad& = &\quad & h^{*}(u) \ct h^{*}(w^{-1})\\
&&& = && h^{*}(u) \ct h^{*}(v^{-1})\ct h^{*}(v) \ct h^{*}(w^{-1})\\
&&& = && h^{*}(u \cdot v^{-1}) \ct h^{*}(v \cdot w^{-1})\\
&&& = && \refl \cdot \refl.
\end{alignedat}
\end{equation*}
if $h^{*}(u \cdot w^{-1}) = h^{*}(w \cdot v^{-1}) = \refl$.
\item[$P$ is closed under inversion]
By $h^{*}(u^{-1}) = h^{*}(u)^{-1} = \refl^{-1} \equiv \refl$ whenever we have
$h^{*}(u) = \refl$.
\item[$P$ holds for the outlines of local confluence diagrams]
This is given directly by $d_2$.
\end{description}
\end{proof}

\subsection{Free $\mathbf{\infty}$-groups}\label{sec:free-groups}

We want to use \cref{thm:dirsetquot} to show that the free higher group $\fhg(M)$ has trivial
fundamental groups.
Recall that this is the example discussed in the introduction, with $\fhg(M)$ defined
in equation \eqref{eq:FA-definition}.

\begin{theorem}\label{thm:fundamental-group-trivial}
 The fundamental groups of the free higher group on a set are trivial.
 In other words, for a set $M$ and any $x : \fhg(M)$,
 we have
 \begin{equation}
  \pi\left(\fhg(M) , x \right) \; = \; \unit.
 \end{equation}
\end{theorem}

We split the proof into several small lemmas.
We keep using the relation $\torel$ of \cref{ex:fg,lem:fg-Noether-NO-confl}.
Further, recall 
the functions $\omega_1$ \eqref{eq:omega1} and $\omega_2$ \eqref{eq:omega2}
from the introduction, as well as the map $\omega$ \eqref{eq:omega-complete}.

\begin{lemma}[{free group; continuing \cref{ex:fg,lem:fg-Noether-NO-confl}}] \label{lem:fg-cycle-condition}
For the relation $\torel$ of \cref{ex:fg}, we can construct the outlines
of a local confluence structure
consisting for each local peak $u : (\ell_x \leadsfrom \ell \leadsto \ell_y)$
of a valley
\begin{equation*}
\overline u : (\ell_x \tostar \ell' \fromstar \ell_y) \text{,}
\end{equation*}
which furthermore can be proven to be coherent by the presence of a 2-path
\begin{equation*}
d_2(u) : \omega_2^*\left(u \cdot \overline{u}^{-1}\right) = \refl \text{.}
\end{equation*}
\end{lemma}
\begin{proof}
We perform a standard critical pair analysis on the span and assume that
$\ell_x$ is obtained from $\ell$ by removing a redex $(x, x^{-1})$, and likewise
that $\ell_y$ is obtained from $\ell$ by removing a redex $(y, y^{-1})$.
Taking in consideration the symmetry of the assumptions we end up with only
three cases:
\begin{enumerate}[ref={(\arabic*)}]
\item \label{enum:freegrp-wb-trivial}
The two redexes are at the same position of $\ell$ (they ``fully overlap''),
implying $x = y$ and $\ell_x = \ell_y$.
\item \label{enum:freegrp-wb-partial}
The two redexes partially overlap, in the sense that $x^{-1} = y$
(or $y^{-1} = x$, which is equivalent).
In this case, we again have $\ell_x = \ell_y$.
\item \label{enum:freegrp-wb-peiffer}
There is no overlap between the two redexes (``Peiffer branching'').

\end{enumerate}
The case \ref{enum:freegrp-wb-trivial} is trivial because we can set
$\wb(u) = \epsilon_{\ell_x}$ and 
\begin{equation*}
\omega_2^*(u) = \omega_2(s) \cdot \omega_2(s)^{-1} = \refl
\end{equation*}
for $s : (\ell \leadsto \ell_x)$.
For the remaining two cases we need to recall the definition of $\omega_1$
and $\omega_2$ and observe the following:
The function $\omega_1 : \List(M \uplus M) \to \fhg(M)$, cf.\ \eqref{eq:omega1},
factors as
\begin{equation} \label{eq:factor-omega1}
\List(M \uplus M) \longrightarrow \List(\mathsf{base} 
= \mathsf{base}) \longrightarrow \mathsf{base} = \mathsf{base}
\end{equation}
where the first map applies $\mathsf{loop}$ on every list element, while the
second concatenates; note that $\fhg(M) \defeq (\mathsf{base} = \mathsf{base})$.
The function $\omega_2$ \eqref{eq:omega2} can then be factored similarly.

For case \ref{enum:freegrp-wb-peiffer} we can remove the redex $(y,y^{-1})$
from $\ell_x$ and have constructed a list $\ell'$ equal to the one we get if we
remove $(x,x^{-1})$ from $\ell_y$.
We combine these reductions to obtain $\wb(u) : (\ell_x \leadsto \ell' \leadsfrom \ell_y)$.
To provide the coherence $d_2(u)$ in this case, we can, by \eqref{eq:factor-omega1},
assume that we are given a list of loops around $\mathsf{base}$ instead of
a list of elements of $M \uplus M$.
We first repeatedly use that associativity of path composition is coherent
(we have ``MacLane's pentagon'' by trivial path induction).
Then, we have to show that the two canonical ways of simplifying
$e_1 \ct (p \ct p^{-1}) \ct e_2 \ct (q \ct q^{-1}) \ct e_3$
to $e_1 \ct e_2 \ct e_3$ are equal.

This can be achieved in two ways:
A common pattern in homotopy type theory is now to generalise to the case that $p$ and $q$ are
equalities with arbitrary endpoints rather than loops, and then do path induction.
If $p$ and $q$ are both $\refl$, then both simplifications become $\refl$ as well.
Instead of applying path induction directly, it is possible to prove this lemma only using
naturality and the Eckmann-Hilton theorem~\cite[Thm 2.1.6]{hott-book}:
The choice of whether to first reduce on the left and then on the right or
vice versa corresponds to the two ways (in the reference called $\star$ and $\star'$)
of defining horizontal composition of 2-paths by first whiskering on the left or
on the right, respectively.
As the proof of the theorem states, these two ways coincide.

In case \ref{enum:freegrp-wb-partial}, we can set $\ell' = \ell_x = \ell_y$
and $\wb(u) = \epsilon_{\ell'}$.
Analogously to case \ref{enum:freegrp-wb-peiffer}, we can construct $d_2(u)$
by showing that the two
ways of reducing $e_1 \ct p \ct p^{-1} \ct p \ct e_2$ to $e_1 \ct p \ct e_2$ are equal.
This time, we have to generalise not only the endpoint of $p$ but
both endpoints of $p$ and the respective endpoints of $e_1$ and $e_2$ to reduce
the problem to the case where $p$ is $\refl$, and the equalities are definitionally
the same.
\end{proof}

\begin{lemma} \label{lem:retract}
 The free higher group $\fhg(M)$ is a retract of $\quotX{\List(M \uplus M)}$, in the sense that there is a map
 \begin{equation}
  \varphi : \fhg(M) \to \quotX{\List(M \uplus M)}
 \end{equation}
 such that $\omega \circ \varphi$ is the identity on $\fhg(M)$.
\end{lemma}
\begin{proof}
 For any $x : M \uplus M$, the operation ``adding $x$ to a list''
 \begin{equation}
  (x \cdot \_) \; : \; \List(M \uplus M) \to \List(M \uplus M)
 \end{equation}
 can be lifted to a function of type
 \begin{equation} \label{eq:endoeqv-on-lists}
  (\quotX{\List(M \uplus M)}) \to (\quotX{\List(M \uplus M)}).
 \end{equation}
 Moreover, the function \eqref{eq:endoeqv-on-lists} is inverse to $(x^{-1} \cons \_)$
 and thus an equivalence.

 Let $\star$ be the unique element of the unit type $\unit$.
 We define the relation $\sim$ on the unit type by $(\star \torel \star) \defeq M$.
 Then, $\mathsf{hcolim} (M \rightrightarrows \unit)$ is by definition the coequaliser 
 $(\quotXY{\unit}{\sim})$, and $\fhg(M)$ is given by $(\iota(\star) = \iota(\star))$.
 This allows us to define $\varphi$ using \cref{thm:lics2019-main} with the constant
 family $P \defeq (\quotX{\List(M \uplus M)})$, with the equivalence of the
 component $e$ given by \eqref{eq:endoeqv-on-lists}.
 
 A further application of \cref{thm:lics2019-main} show that $\omega \circ \varphi$
 is pointwise equal to the identity.
\end{proof}

\begin{proof}[Proof of \cref{thm:fundamental-group-trivial}]
 By \cite[Thms 7.2.9 and 7.3.12]{hott-book}, the statement of the theorem is
 equivalent to the claim that $\trunc 1 {\fhg(M)}$ is a set.
 
 We now consider the following diagram:
 
\begin{equation} \label{eq:diagram}
\begin{tikzpicture}[x=3.2cm,y=-2.5cm,baseline=(current bounding box.center)]
  \tikzset{arrow/.style={shorten >=0.1cm,shorten <=.1cm,-latex}}
\node (FA1) at (0,0) {$\fhg(M)$}; 
\node (Lquot) at (0,1) {$\quotX{\List(M \uplus M)}$}; 
\node (Lsquot) at (0,2) {$\quotX{\List(M \uplus M)}$}; 
\node (FA2) at (1,1) {$\fhg(M)$}; 
\node (FAT) at (2,1) {$\trunc 1 {\fhg(M)}$}; 

\draw[arrow] (FA1) to node [right] {$\varphi$} (Lquot);
\draw[arrow] (FA1) to node [above right] {$\tproj 1 -$} (FAT);
\draw[arrow] (Lquot) to node [right] {$\tproj 0 -$} (Lsquot);
\draw[arrow] (Lquot) to node [above] {$\omega$} (FA2);
\draw[arrow] (FA2) to node [above] {$\tproj 1 -$} (FAT);
\draw[arrow,dashed] (Lsquot) to node [below right] {} (FAT);
\end{tikzpicture}
\end{equation}
 The dashed map exists by 
 the combination of \cref{thm:dirsetquot} (note that we are in the simplified case
 where the type to be quotiented is a set) together with \cref{lem:fg-cycle-condition}
 (and \cref{lem:simu-is-consec}).
 By construction, the bottom triangle commutes.
 The top triangle commutes by \cref{lem:retract}.
 
 Therefore, the map $\tproj 1 -$ factors through a set (namely $\quotX{\List(M \uplus M)}$).
 This means that $\trunc 1 {\fhg(M)}$ is a retract of a set, and therefore itself a set.
\end{proof}

\subsection{Pushouts of 1-types}\label{sec:pushouts}

In this section, we will prove another theorem using the characterisation of
maps out of quotients \cref{thm:dirsetquot}.
At first glance it might look completely distinct from the application of free
$\infty$-groups, but, as we will see in \cref{sec:special}, it is a more
general formulation of the same phenomenon.
The subject of study of this section will be \emph{pushouts} of types.
We will always assume that we are given a span of 
$B \xleftarrow f A \xrightarrow g C$ of types and functions, of which we will
take the pushout:
  \begin{equation} \label{eq:pushout}
  \begin{tikzpicture}[x=1.5cm,y=-1.0cm,baseline=(current bounding box.center)]
   \node (A) at (0,0) {$A$};
   \node (C) at (1,0) {$C$};
   \node (B) at (0,1) {$B$};
   \node (D) at (1,1) {$B \sqcup_A C$};
  
   \draw[->] (A) to node [left] {\scriptsize $f$} (B);
   \draw[->] (A) to node [above] {\scriptsize $g$} (C);
   \draw[->] (B) to node [above] {\scriptsize $i_0$} (D);
   \draw[->] (C) to node [left] {\scriptsize $i_1$} (D);
  \end{tikzpicture}
 \end{equation}
The pushout is, as common in homotopy type theory, defined as a higher inductive
type with point constructors $i_0 : B \to B \sqcup_A C$ and $i_1  : C \to B \sqcup_A C$
as well as a path constructor
\begin{equation*}
\glue : \Pi(a : A).\, i_0(f(a)) = i_1(g(a)) \text{,}
\end{equation*}
which makes the diagram \eqref{eq:pushout} commute.

In constrast to \cref{sec:free-groups}, we will not prove statement about
first but about \emph{second} homotopy groups, but again,
we do not make any statements about the homotopy levels above that.

\begin{theorem} \label{thm:pushout-1-type}
 Given a pushout as in \eqref{eq:pushout}, if $A$ is a set and $B$, $C$
 are $1$-types, then all second homotopy groups of $B\sqcup_A C$ are trivial.
 In other words, $\trunc 2 {B\sqcup_A C}$ is a $1$-type.
\end{theorem}

\begin{proof}[Proof sketch]
The argument is almost completely analogous to the proof of
\cref{thm:fundamental-group-trivial}.
The main difference is that the type $\List(M \uplus M)$ is not sufficient any more.
Instead, we need to be slightly more subtle when we encode the equalities
in the pushout.
The following construction is due to Favonia and Shulman \cite{favonia:SvK},
who use it in their formulation of the \emph{Seifert-van Kampen Theorem}.

Given the square in \eqref{eq:pushout} and $b, b' : B$, we consider the type
$L_{b,b'}$ of lists of the form
\begin{equation} \label{eq:SvK-lists}
[b, p_0, x_1, q_1, y_1, p_1, x_2, q_2, y_2, \ldots, y_n, p_n, b'] 
\end{equation}
\ where\footnote{We remove the $0$-truncations around the path spaces.
These are without effect here since $B$, $C$ are $1$-types.}
\begin{alignat}{3}
& x_i : A & \qquad &  y_i : A \\
& p_0 : b = f(x_1) && p_n : f(y_n) = b' \\
& p_i : f(y_i) = f(x_{i+1}) && q_i : g(x_i) = g(y_i)
\end{alignat}
The corresponding relation is generated by
\begin{equation} \label{eq:SvK-relation}
\begin{aligned}
  [\ldots, q_k, y_k, \refl, y_k, q_{k+1}, \ldots] \; & \leadsto \; [\ldots, q_k \ct q_{k+1}, \ldots] \\
  [\ldots, p_k, x_k, \refl, x_k, p_{k+1}, \ldots] \; & \leadsto \; [\ldots, p_k \ct p_{k+1}, \ldots] 
\end{aligned}
\end{equation}
The statement of the Seifert-van Kampen Theorem is that the set-quotient
$\setquotX{L_{b,b'}}$
is equivalent to the set-truncated type
$\trunc 0 {i_0(b) = i_0(b')}$ of equalities in the pushout.
Similarly to $L_{b,b'}$, there are three further types of lists where one or
both of the endpoints are in $C$ instead of $B$.
In general, we can define a type of lists $L_{x,x'}$ for $x,x' : B \uplus C$, and the
Seifert-van Kampen Theorem states that $\setquotX{L_{x,x'}}$ is equivalent
to $\trunc 0 {i(x) = i(x')}$, with $i : B \uplus C \to B \sqcup_A C$ given by $(i_0, i_1)$.

The construction of $\omega$ and $\varphi$ is essentially the same as before,
using the version of \cref{thm:lics2019-main} for pushouts available in \cite{KrausVonRaumer_pathSpaces}.
For the relation \eqref{eq:SvK-relation}, we can show the analogous to
\cref{lem:fg-Noether-NO-confl,lem:fg-cycle-condition}.
The analogous to \eqref{eq:diagram} is
\begin{equation} \label{eq:diagram2}
\begin{tikzpicture}[x=3.7cm,y=-2.5cm,baseline=(current bounding box.center)]
  \tikzset{arrow/.style={shorten >=0.1cm,shorten <=.1cm,-latex}}
\node (FA1) at (0,0) {$i(x) = i(x')$}; 
\node (Lquot) at (0,1) {$\quotX{L_{x,x'}}$}; 
\node (Lsquot) at (0,2) {$\setquotX{L_{x,x'}}$}; 
\node (FA2) at (1,1) {$i(x) = i(x')$}; 
\node (FAT) at (2,1) {$\trunc 1 {i(x) = i(x')}$}; 

\draw[arrow] (FA1) to node [right] {$\varphi$} (Lquot);
\draw[arrow] (FA1) to node [above right] {$\tproj 1 -$} (FAT);
\draw[arrow] (Lquot) to node [right] {$\tproj 0 -$} (Lsquot);
\draw[arrow] (Lquot) to node [above] {$\omega$} (FA2);
\draw[arrow] (FA2) to node [above] {$\tproj 1 -$} (FAT);
\draw[arrow,dashed] (Lsquot) to node [below right] {} (FAT);
\end{tikzpicture}
\end{equation}
There is a small subtlety: Since $A$ is a set and $B$, $C$ are $1$-types,
the type of lists $L_{x,x'}$ is a set.
This is important since it allows us (as before) to use the simpler version
of \cref{thm:dirsetquot}.
The above diagram shows that $\trunc 1 {i(x) = i(x')}$ is a set.
Choosing $x$ and $x'$ to be identical, this means that $\trunc 1 {\Omega(B\sqcup_A C,i(x))}$
is a set, which is equivalent to the statement that $\trunc 0 {\Omega^2(B\sqcup_A C,i(x))}$
(the second homotopy group) is trivial.
It follows by the usual induction principle of the pushout that 
$\trunc 0 {\Omega^2(B\sqcup_A C,z)}$ for arbitrary $z : B\sqcup_A C$ is trivial.
\end{proof}

\subsection{Relation between the applications}\label{sec:special}

Let us compare the lists used in the proofs of \cref{thm:fundamental-group-trivial}
and \cref{thm:pushout-1-type}.
We can observe that the former ones are a specialisation of the latter once with
the following restrictions:
The types $B$ and $C$ are each set to be the unit type $\unit$.
$A$ in \eqref{eq:pushout} is set to be $A' \uplus \unit$ when we want to consider the
free group on $A'$.
Then, the elements $b$, $b'$, $p_i$, and $q_i$ in \eqref{eq:SvK-lists} carry
no information, and a list entry of the right summand of $A' \uplus \unit$ indicates
a switch from the ``left to right'' or vice versa in $\List(A' \uplus A')$.
Under this transformation, the relations \eqref{eq:fg-relation} and
\eqref{eq:SvK-relation} correspond to each other.

Indeed, all of the following questions can be reduced to the question, whether
the pushout $B \sqcup_A C$ of 1-types $B$ and $C$ over a set $A$ is a 1-type.
The first one is the problem which we approximated in \cref{sec:free-groups}:
\begin{enumerate}[(i)]
\item \label{item:1}
Is the free higher group on a set again a set?
\item \label{item:2}
Is the suspension of a set a $1$-type (open problem recorded in \cite[Ex 8.2]{hott-book})?
\item \label{item:3}
Given a $1$-type $B$ with a base point $b_0 : B$. If we add a single loop around $b_0$, it the type still a $1$-type?
\item \label{item:4}
Given $B$ and $b_0$ as above, imagine we add $M$-many loops around $b_0$ for some given set $M$. Is the resulting type still a $1$-type?
\item \label{item:5}
If we add a path (not necessarily a loop) to a $1$-type $B$, is the result still a $1$-type?
\item \label{item:6}
If we add an $M$-indexed family of paths to a $1$-type $B$ (for some set $M$), is the resulting type still a $1$-type?
\end{enumerate}

All questions are of the form:
\begin{quote}
 ``Can a change at level $1$
   induce a change at level $2$ or higher?''
\end{quote}
Only \ref{item:1} seems to be about level $0$ and $1$, but this is simply because we have taken a loop space.
With our \cref{thm:dirsetquot}, we can show an approximation for each of these questions analogously to \cref{thm:fundamental-group-trivial}.
This means that we show:
\begin{quote}
 ``A change at level $1$ does not induce a change at level $2$ (but we don't know about higher levels).''
\end{quote}
We can obtain all of these approximations by setting in \cref{thm:pushout-1-type}, respectively:
\begin{enumerate}[(i)]
 \item $B$, $C$ to both be the unit type $\unit$ and $A$ is $A' \uplus \unit$,
  where $A'$ is the set on which we want the free higher group
  (this is the usual translation from coequalisers to pushouts); 
 \item $B$ and $C$ both to be $\unit$;
 \item $A \defeq \unit$ and $C$ be the circle $\mathsf{S}^1$;
 \item $A \defeq \unit$ and $C \defeq M \times \mathsf{S}^1$;
 \item $A$ to be the 2-element type $\bool$ and $C \defeq \unit$;
 \item $A \defeq M \times \bool$ and $C \defeq M$.
\end{enumerate}

\section{Concluding remarks}

Our work has shown that methods from higher-dimensional rewriting can be used
to tackle some of the coherence problems appearing in homotopy type theory.
One limitation of our results so far is that they only make statements about one
specific dimension of the spaces which we consider.
It may very well be possible to generalise our method to show ``higher'' versions
of the same coherences, or, in other words,
better approximations of the same open problems.
For example, one could try to relax the condition of 1-truncatedness in \cref{thm:dirsetquot}
to 2-truncatedness.
For this generalisation, we expect that the proofs of 2-dimensional coherence would have to be
coherent as well.
It remains to see whether 3-dimensional rewriting theory, as proposed
by Mimram \cite{t3rt}, could be a useful vantage point to guide proofs
about these 3-dimensional coherence theorems.

The ``fully untruncated'' version of \cref{thm:pushout-1-type} would state that the pushout of $B \leftarrow A \to C$ (with $A$ a set and $B$, $C$ groupoids) is a groupoid, with the special case being that the free higher group on a set is a set.
We do not expect that this is provable in homotopy type theory.
Even if a workable version of $\infty$-rewriting theory is formulated, this looks like one of the problems requiring an infinite tower of coherences that, akin to semisimplicial types~\cite{herbelin:sst,Kraus:theroleofSST}, are not expected to be expressible in homotopy type theory. However, we conjecture that this can be done in \emph{two-level type theory} \cite{voe:hts,alt-cap-kra:two-level,ann-cap-kra:two-level} in the style of \cite{kraus:inftyCwFs}.
At the workshop \emph{Logique et Structures Sup\'erieures} at the Centre International de Rencontres Math\'ematiques (CIRM)
in February 2022, Christian Sattler outlined an argument to generalise the statement to any externally chosen truncation level \cite{sattler:cirm-abstract}.

Another line of research about potential generalisation is the question of whether
it is possible to weaken the assumption that the polygraph is terminating.
Instead, it could be enough to assume the \emph{decreasingness} of the relation,
which Vincent van Oostrom suggests as an alternative \cite{vincent-08}.

\subsection*{Acknowledgments}
We would first of all like to thank Vincent van Oostrom for the interest he took
in reading our LICS paper \cite{krausVonRaumer:wellfounded} and for the detailed
explanation on how our approach compared to the terminology and the results about
abstract rewriting systems.
The more modular presentation of our proofs, the generalisation of the results, and the clearer connection to higher dimensional rewriting would not have happened without his input.

We would also like to thank Carlo Angiuli who, like Vincent van Oostrom, pointed
us towards the work of Craig C. Squier on rewriting systems.

Furthermore, we would like to express our gratitude to many people whose comments
helped us to improve the presentation of both our LICS and the current paper.
We thank the Theory of Computation group at the University of Birmingham, the Type Theory group in Budapest, the Functional Programming lab in Nottingham, the participants and the organiser (Chuangjie Xu) of the \emph{FAUM} meeting (Herrsching 2019) and \emph{Types in Munich} (online 2020), 
the LICS'20 organizers, as well as the participants and the organisers (Dan Christensen and Chris Kapulkin) of \emph{HoTTEST} (Homotopy Type Theory Electronic Seminar Talks).
We especially acknowledge the remarks by Steve Awodey, Ulrik Buchholtz, Thierry Coquand, Eric Finster, Mart{\'i}n H{\"o}tzel Escard{\'o}, Egbert Rijke, Anders M\"ortberg, Chuangjie Xu and, in particular, Christian Sattler.

\bibliographystyle{plain}
\bibliography{master,rewriting_refs}

\end{document}